\documentclass[11pt,letterpaper]{article}
\pdfoutput=1
\usepackage[utf8]{inputenc}
\usepackage[T1]{fontenc}
\usepackage[margin=1in]{geometry}
\usepackage{amsmath,amsthm,amssymb,thmtools,thm-restate,booktabs,color,doi,graphicx,latexsym,url,xcolor,xspace}
\usepackage{todonotes}
\usepackage[numbers,sort&compress]{natbib}
\usepackage{microtype,hyperref}
\usepackage[inline]{enumitem}
\setlist[itemize]{label=--}
\setlist[enumerate]{label=(\arabic*),labelindent=\parindent,leftmargin=*}

\usepackage{sectsty}
\allsectionsfont{\boldmath}

\definecolor{citecolor}{HTML}{0000C0}
\definecolor{urlcolor}{HTML}{000080}

\hypersetup{
    colorlinks=true,
    linkcolor=black,
    citecolor=citecolor,
    filecolor=black,
    urlcolor=urlcolor,
    pdftitle={}, pdfauthor={} }

\newtheorem{theorem}{Theorem}
\newtheorem{lemma}[theorem]{Lemma}

\newenvironment{myabstract}
{\list{}{\listparindent 1.5em\itemindent    \listparindent
		\leftmargin    1cm
		\rightmargin   1cm
		\parsep        0pt}\item\relax}
{\endlist}

\newenvironment{mycover}
{\list{}{\listparindent 0pt
		\itemindent    \listparindent
		\leftmargin    1cm
		\rightmargin   1cm
		\parsep        0pt}\raggedright
	\item\relax}
{\endlist}

\newcommand{\myemail}[1]{\,$\cdot$\, {\small #1}}
\newcommand{\myaff}[1]{{\small #1}\par\medskip}

\renewcommand{\vec}[1]{\mathbf{#1}}

\DeclareMathOperator*{\E}{\mathbf{E}}
\DeclareMathOperator*{\Var}{\mathbf{Var}}

\DeclareMathOperator*{\polylog}{polylog}

\newcommand{\N}{\mathbb{N}}

\newcommand{\namedref}[2]{\hyperref[#2]{#1~\ref*{#2}}}
\newcommand{\sectionref}[1]{\namedref{Section}{#1}}

\newcommand{\tableref}[1]{\namedref{Table}{#1}}
\newcommand{\equationref}[1]{\hyperref[#1]{Eq~(\ref*{#1})}}
\newcommand{\theoremref}[1]{\hyperref[#1]{Theorem~\ref*{#1}}}
\newcommand{\lemmaref}[1]{\hyperref[#1]{Lemma~\ref*{#1}}}
\newcommand{\noteref}[1]{\hyperref[#1]{note~\ref*{#1}}}
\newcommand{\appendixref}[1]{\hyperref[#1]{Appendix~\ref*{#1}}}
\newcommand{\corollaryref}[1]{\hyperref[#1]{Corollary~\ref*{#1}}}

\begin{document}

\begin{mycover}
	{\huge\bfseries\boldmath Majority consensus thresholds in competitive Lotka--Volterra populations \par}

	\bigskip
	\bigskip

\textbf{Matthias F{\"u}gger}
\myemail{mfuegger@lmf.cnrs.fr} \\
\myaff{Universit{\'e} Paris-Saclay, CNRS, ENS Paris-Saclay, LMF}

\textbf{Thomas Nowak}
\myemail{thomas@thomasnowak.net} \\
\myaff{Universit{\'e} Paris-Saclay, CNRS, ENS Paris-Saclay, LMF \& Institut Universitaire de France}

\textbf{Joel Rybicki}
\myemail{joel.rybicki@hu-berlin.de} \\
\myaff{Humboldt University of Berlin}

\bigskip
\end{mycover}

\medskip
\begin{myabstract}
  \noindent\textbf{Abstract.}
One of the key challenges in synthetic biology is devising robust signaling primitives for engineered microbial consortia. In such systems, a fundamental signal amplification problem  is the majority consensus problem: given a system with two input species with initial difference of $\Delta$ in population sizes, what is the probability that the system reaches a state in which only the initial majority species is present?

In this work, we consider a discrete and stochastic version of competitive Lotka--Volterra dynamics, a standard model of microbial community dynamics. We identify new threshold properties for majority consensus under different types of interference competition:
\begin{itemize}
\item We show that under so-called self-destructive interference competition between the two input species,
  majority consensus can be reached with high probability if the initial difference satisfies $\Delta \in \Omega(\log^2 n)$, where $n$ is the initial population size. This gives an \emph{exponential} improvement compared to the previously known bound of $\Omega(\sqrt{n \log n})$ by Cho et al. [Distributed Computing, 2021] given for a special case of the competitive Lotka--Volterra model.
  In contrast, we show that an initial gap of $\Delta \in \Omega(\sqrt{\log n})$ is necessary.

  \item On the other hand, we prove that under \emph{non-self-destructive} interference competition, an initial gap of $\Omega(\sqrt{n})$ is necessary to succeed with high probability and that a $\Omega(\sqrt{n \log n})$ gap is sufficient.
    \end{itemize}

This shows a strong qualitative gap between the performance of self-destructive and non-self-destructive interference competition. Moreover, we show that if in addition the populations exhibit interference competition between the individuals of the \emph{same} species, then majority consensus cannot always be solved with high probability, no matter what the difference in the initial population counts.
\end{myabstract}

\thispagestyle{empty}
\setcounter{page}{0}
\newpage

\section{Introduction}

Synthetic biology is a discipline focusing on the rational engineering of biological systems~\cite{brenner2008engineering,cameron2014brief}. Early results built basic computational modules, such as memory~\cite{gardner2000construction}, clocks~\cite{elowitz2000synthetic,danino2010synchronized}, and sensors~\cite{levskaya2005engineering} in the bacterium \textit{Escherichia coli}. Recently, synthetic biologists have started to engineer \emph{synthetic consortia} consisting of multiple interacting microbial species that collectively implement \emph{distributed} biological circuits~\cite{tamsir2011robust,regot2011distributed,macia2012distributed,bittihn2018rational}.
This has lead the bioengineering community to face a fundamental challenge of distributed systems: the need for robust \emph{coordination primitives} to coordinate the activities of all different microbial populations comprising the  circuit~\cite{alnahhas2020majority,macia2012distributed,bittihn2018rational,grandel2021control}.

Synthetic biologists have recognized that (a) such problems are studied in the field of distributed computing and that (b) the existing models of distributed computing poorly capture key aspects of microbial systems~\cite{benenson2012biomolecular,grozinger2019pathways,karkaria2020microbial,del2018future}. Classical models of distributed computing ignore even the most elementary \emph{ecological processes} that take place in microbial consortia, such as the stochastic reproduction and mortality of individual cells and competition between different species~\cite{gonze2018microbial,zhou2017stochastic}.

This limitation holds even for many popular models of distributed computing explicitly motivated by biological computation, such as the widely studied population protocol model~\cite{angluin2008simple,alistarh2017time,berenbrink2020optimal,doty2022time,bankhamer2022population}.
Unfortunately, these fundamental ecological processes are (1) in general unavoidable in synthetic microbial consortia, (2) lead to stochastic fluctuations in the community size and composition, which are  key drivers of microbial population dynamics~\cite{zhou2017stochastic}, and importantly, (3) dealing with them is one of the main challenges in engineering microbial consortia~\cite{bittihn2018rational,grandel2021control,karkaria2020microbial}.

While theoretical foundations of molecular and biological computation have gained increasing attention in recent years~\cite{fisher2011biology,chazelle2012natural,doty2012theory,feinerman2013theoretical,navlakha15biological,alistarh-survey,rashid2021epigenetic}, there is scarcely any work focusing on the theory of distributed computing in \emph{microbial consortia}, despite the fact that distributed computing has become a fundamental paradigm in synthetic biology.
Here, we initiate the complexity theoretic study of the \emph{computational power of programmable ecological interactions}~\cite{li2022synthetic} in microbial circuits. 

\subsection{The problem: majority consensus}

Majority consensus is a fundamental problem  in distributed computing~\cite{angluin2008simple,draief2012-convergence,cooper2014power,mertzios2014majority,feinerman2017breathe,alistarh2018space-optimal,doty2022time,condon2020approximate}. In this problem, each node in the system is given a local input bit, and the task is for each node to output the input bit given to the majority of the nodes.

\emph{Majority consensus} has also been identified as a useful signaling primitive for engineered microbial consortia in theory~\cite{cho2021distributed} and practice~\cite{alnahhas2020majority}.
The goal is to design a genetic circuit, where the (gene expression) state of each cell is determined by the species which is in the majority~\cite{alnahhas2020majority}.
This primitive can be used as a robust differential signal amplifier~\cite{cho2021distributed} when composing complex circuits implemented by different populations~\cite{regot2011distributed}.
However, it remains unclear how to efficiently solve this problem in the microbial setting, where individuals replicate, die, and compete with one another.

The performance of majority consensus protocols is typically studied as a function of the initial population size $n$ and the initial difference $\Delta$ between the counts of the majority and minority input species.
In \emph{exact} majority consensus, the goal is to correctly compute the majority with probability~1 for any $\Delta>0$, whereas
in \emph{approximate} majority consensus, the protocols are allowed to fail with a probability that depends on $n$ and $\Delta$.
Typically, the aim is to obtain protocols that succeed with high probability in $n$ for as small as possible $\Delta$.
Intuitively, the smaller $\Delta$ a protocol can deal with, the better the protocol tolerates noise:
in the synthetic biology setting, such protocols are instrumental in amplifying signals produced by noisy (biological) processes, e.g.,  microbial subcircuits.

For example, in the stochastic population protocol model~\cite{angluin2006computation}, where a random scheduler picks pairs of nodes to interact uniformly at random in each time step, both exact~\cite{alistarh2015-fast,draief2012-convergence,alistarh2017time,mertzios2014majority,doty2022time} and approximate~\cite{angluin2008simple,czyzowicz2022convergence} versions of majority consensus have been studied.
There is a simple (but challenging to analyze) 3-state protocol~\cite{angluin2008simple} that solves approximate majority consensus in $O(n \log n)$ interactions with high probability whenever $\Delta \in \Omega(\sqrt{n} \cdot \log n)$. This  protocol can in principle implement a cell cycle switch in cellular populations~\cite{cardelli2012cell}.
In contrast, for \emph{exact} majority consensus, it is known that any $O(1)$-state protocol requires $\Omega(n^2)$ expected interactions~\cite{alistarh2017time}, but
there is a $O(\log n)$-state protocol that solves the problem in $O(n \log n)$ interactions in expectation~\cite{doty2022time}.

\subsection{Our focus: from molecular to microbial computation}
The population protocol model captures uncertainty arising from unpredictable \emph{interaction patterns}, but it does not capture \emph{demographic noise}, i.e., the stochastic fluctuations in the community composition that arise from the chance events of the underlying  biomolecular~\cite{wilkinson2018stochastic,benenson2012biomolecular} and ecological processes~\cite{zhou2017stochastic,gonze2018microbial,lande2003stochastic} taking place in microbial systems. For small  populations, as in the case of synthetic microbial consortia, demographic noise is known to have a significant impact on the realized dynamics.

In this work, we take steps towards developing a theory of microbial computation that investigates the \emph{computational power of ecological interactions}.
Motivated by the recent experimental~\cite{alnahhas2020majority,li2022synthetic} and mathematical modeling work utilizing ecological competitive mechanisms~\cite{mao2015slow,cho2021distributed,andaur2021reaching} in the design of distributed majority consensus algorithms in synthetic microbial consortia, we study the following question at the interface of theoretical computer science and synthetic ecology:
\begin{quote}
How does (a) demographic noise and (b) the choice of engineered, competitive mechanisms impact the performance of microbial majority consensus protocols?
\end{quote}
We investigate the  majority consensus problem in discrete Lotka--Volterra models of well-mixed, competitive microbial communities. These are standard models used in microbiology that explicitly model the reproductive and competitive dynamics of microbial species over time~\cite{gonze2018microbial,zhou2017stochastic,dedrick2023does}.
The analysis of \emph{stochastic} population models is considered important but highly challenging in the biomathematics~\cite{jagers2010plea,geritz2012mathematical,black2012stochastic} and statistical physics communities~\cite{mckane2004stochastic,huang2015stochastic,reichenbach2006coexistence,dobrinevski2012extinction}. Recently, their rigorous analysis has started to gain interest also in theoretical computer science ~\cite{cho2021distributed,czyzowicz2022convergence,andaur2021reaching}.

We develop new techniques to analyze the behavior of majority consensus dynamics in stochastic, competitive Lotka--Volterra models, and show that in principle (1) ecological processes can be \emph{algorithmically exploited} to obtain robust majority consensus protocols, and that (2) different ecological mechanisms can have substantial impact on the performance of microbial protocols.

\subsection{The model: stochastic Lotka--Volterra dynamics}\label{ssec:lv}

We now introduce the stochastic two-species Lotka--Volterra (LV) models of competing microbial populations.
These LV models generalize the models of microbial majority consensus previously studied in the distributed computing community~\cite{cho2021distributed,andaur2021reaching}.
Lotka--Volterra dynamics can be derived mechanistically from several different biological assumptions. Here, we have derived the models assuming interference competition~\cite{granato2019evolution} between the species, as interference competition can be programmed into synthetic microbial populations using readily available genetic modules~\cite{li2022synthetic}.

\paragraph{The Lotka--Volterra models.}
As typical for non-spatial biological population models~\cite{black2012stochastic,wilkinson2018stochastic}, our models are formally represented as chemical reaction networks~\cite{gillespie1977exact,Bre99}, assuming unit volume. For the sake of brevity, we do not give the full general definition of the chemical reaction network formalism, as we only work with two basic models.
Following the standard reaction kinetics notation, the Lotka--Volterra model with \emph{self-destructive} interference competition is given by
\begin{equation}\label{eq:sd-model}
X_i \xrightarrow{\beta} X_i + X_i \qquad X_i \xrightarrow{\delta} \emptyset \qquad X_{i} + X_{1-i} \xrightarrow{\alpha_i} \emptyset \qquad X_i + X_{i}  \xrightarrow{\gamma_i} \emptyset,
\end{equation}
where $X_i$ denotes an individual of species $i \in \{0,1\}$,  $\emptyset$ denotes the removal of each reactant species and $\alpha_i,\beta,\delta,\gamma_i \ge 0$ are constant rate parameters. We note that here we treat the reactions with reactants $X_0 + X_1$ and $X_1 + X_0$ of different species formally as different reactions unlike typically in chemical reaction network formalism -- this purely for notational convenience.

Our second model is the Lotka--Volterra model with \emph{non-self-destructive} competition given by \begin{equation}\label{eq:nsd-model}
X_i \xrightarrow{\beta} X_i + X_i \qquad X_i \xrightarrow{\delta} \emptyset \qquad X_{i} + X_{1-i} \xrightarrow{\alpha_i} X_i \qquad X_{i} + X_{i} \xrightarrow{\gamma_i} X_i,
\end{equation}
for $i \in \{0,1\}$. We often write $\alpha=\alpha_0 + \alpha_1$ and $\gamma = \gamma_0 + \gamma_1$. We say that the system is \emph{neutral} if both species have identical rate parameters. Unless otherwise specified, we consider neutral communities. In both models, the reactions have at most two reactants. We call reactions with a single reactant \emph{individual reactions} and reactions with two reactants \emph{pairwise interactions} between individuals.

\paragraph{Biological interpretation.}
We assume a well-mixed setting as encountered in a bioreactor, with a growing microbial community consisting of two populations of different species. We primarily focus on the early stages of the microbial population dynamics, the so-called exponential phase of microbial growth, where the population size is far from the carrying capacity of the environment (e.g., growth is not limited by availability of nutrients or space).

Each individual of input species $i \in \{0,1\}$ reproduces at per-capita rate $\beta \ge 0$ and dies at per-capita rate $\delta \ge 0$. In addition to these reproductive dynamics, we assume that interactions between individuals happen via interference competition: individuals of  species $i$ encounter and individuals of species $1-i$ at rate $\alpha_i \ge 0$, and individuals of the same species $i$ encounter and kill each other at rate $\gamma_i \ge 0$. Competition between individuals of different species is called \emph{interspecific}, whereas competition between individuals of the same species is \emph{intraspecific}.

When species engage in a competitive interaction, the outcome can be either symmetric (both individuals die) or asymmetric (only one individual dies). These two cases correspond to the two different models (1) and (2), and they have biologically very distinct interpretations. The former scenario corresponds to \emph{self-destructive interference} (e.g., cells release a bacteriocin via lysis), whereas the latter scenario corresponds to \emph{non-self-destructive interference} (e.g., cells secrete a bacteriocin or puncture the membrane of other cells on  physical contact). Both types of competition are exhibited by bacterial species~\cite{granato2019evolution}. Non-self-destructive interference can also be implemented in engineered populations using, e.g., programmable plasmid conjugation~\cite{marken2023addressable}, as suggested by Cho et al.~\cite{cho2021distributed}.

\paragraph{Stochastic kinetics.}
For both model variants, a \emph{configuration} is a vector $\vec x  = (x_0, x_1) \in \mathbb{N}^2$, where $x_i$ gives the count of species $i$ in the configuration. In a configuration $(x_0,x_1)$, the individual birth reactions of species $i$ have propensity $\beta x_i$ and individual death reactions have propensity $\delta x_i$. The propensity of an interspecific competition reaction is $\alpha x_0 x_1$ and the propensity of an intraspecific competition reaction is $\gamma x_i (x_i-1)/2$. The total propensity of the configuration $(x_0,x_1)$ is given by
\[
\varphi(x_0,x_1) = \sum_{i \in \{0,1\}} \left( \alpha_i x_0x_1 + \beta x_i + \delta x_i + \gamma_i x_i(x_i-1)/2 \right).
\]
Under continuous-time stochastic kinetics, the time until the next reaction in a configuration $\vec x$ is exponentially distributed with rate $\varphi(\vec x)$. Given that the current configuration is $\vec x$, the probability that $R$ is the next reaction is  $\varphi_R(\vec x)/\varphi(\vec x)$, where $\varphi_R(\vec x)$ is the propensity of the reaction (as above).

The \emph{stochastic trajectory} under stochastic kinetics is given by a continuous-time Markov process $\vec X = (\vec X_t)_{t \ge 0}$ on the state space $\mathbb{N}^2$, where $Q(\vec x,\vec y)$ gives the propensity of the chain from transitioning from state $\vec x$ to state $\vec y$. In this work, we primarily focus on the \emph{jump chain} $\vec S = (\vec S_t)_{t \ge 0}$ of $\vec X$, which is the discrete-time Markov chain given by the transition probabilities $P(\vec x,\vec y) = Q(\vec x,\vec y) / \varphi(\vec x)$ for $\varphi(\vec x) > 0$. If $\varphi(\vec x) = 0$, then $P(\vec x, \vec x) = 1$ and $P(\vec x,\vec y)=0$ for $y \neq x$. Here, $\vec S_t$ represents the counts of both species after $t \in \mathbb{N}$ reactions have occurred.

\paragraph{Majority consensus.}
We say that species $i \in \{0,1\}$ in a state $\vec S_t = (x_0, x_1)$ is the \emph{majority species} if $x_i > x_{1-i}$ holds. In particular, we say that the majority species in the initial configuration $\vec S_0$ is the \emph{initial majority species}. The species who is not the initial majority species is the \emph{initial minority species}.
We say that a configuration $(x_0,x_1)$ has \emph{reached consensus} if $x_0 = 0$ or $x_1 = 0$. In such a configuration, we say that the species $i$ has won if $x_i > 0$ and $x_{1-i} = 0$. i.e., the species $i$ is the majority species in the configuration $(x_0,x_1)$.
We define the \emph{consensus time} of the chain $\vec S$ to be
\[
T(\vec S) = \inf \{ t : \vec S_t \textrm{ has reached consensus} \}
\]
the minimum time until (at least) one of the species goes extinct.
We say that the chain $\vec S$ \emph{reaches majority consensus} if $T(\vec S)$ is finite and the initial majority species has positive count at time $T(\vec S)$. We define $\rho(\vec S)$ to be the probability that the chain $\vec S$ reaches majority consensus.

Without loss of generality, we assume throughout that the first species is the initial majority species, i.e., $\vec S_0 = (a,b)$ for some $a > b > 0$. We use $n=a+b$ to denote the size of the total initial population. For every time step $t \ge 0$, we define $\Delta_t = S_{t,0} - S_{t,1}$ to be the difference between the counts of the initial majority and the initial minority species. In particular, $\Delta_0$ gives the initial gap between the initial majority and minority species.

\subsection{Our contributions}

In this work, we are interested in how the probability $\rho(\vec S)$ of majority consensus behaves as a function of the initial gap $\Delta_0$ under stochastic kinetics. To this end, we say that $\Psi(n) \ge 0$ is a \emph{majority consensus threshold} for a Lotka--Volterra model if $\rho(\vec S) \ge 1-1/n$ if and only if $\Delta_0 \ge \Psi$. We identify the asymptotic (and sometimes exact) majority consensus thresholds for two-species competitive Lotka--Volterra systems under different modes of inference competition. Our main results are summarized in \tableref{tab:comparison} and are as follows.

\begin{table}[]
\begin{tabular}{@{}l|c|c|l@{}}
\toprule
\textbf{Competition} & \textbf{Self-destructive}                       & \textbf{Non-self-destructive}                       & \textbf{Reference} \\
\midrule
Interspecific only & $\Omega(\sqrt{\log n})$ --- $O(\log^2n)$ & $\Omega(\sqrt{n})$ --- $O(\sqrt{n \log n})$ & Sec.~\ref{sec:self-destructive} and Sec.~\ref{sec:non-self-destructive} \\
Both inter- and intraspecific & $\ge n -1$ & $\ge n - 1$  & Sec.~\ref{sec:both-comp}  \\
Intraspecific only & $\infty$ & $\infty$ & Sec.~\ref{sec:intra-only}  \\
\midrule
Interspecific and $\delta=0$ & $O(\sqrt{n \log n})$ & $O(\sqrt{n\log n})$ (*) & \cite{cho2021distributed} and \cite{andaur2021reaching} \\
No competition ($\alpha=\gamma=0$)& $n-1$ & $n-1$ & \cite{andaur2021reaching} \\
\bottomrule
\end{tabular}
\caption{Worst-case majority consensus thresholds for different LV models. The first three rows are new results, the last two results are from prior work.
  (*) The model of~\cite{andaur2021reaching} is not strictly a special case of the LV model, as it assumes bounded, non-mass-action birth rates.}
\label{tab:comparison}
\end{table}

\paragraph{Interspecific competition.}
In systems with interspecific competition and no intraspecific competition (i.e., the case $\alpha>0$ and $\gamma=0$) we show the following results.
First, for self-destructive interference competition, we show that the majority consensus threshold lies in a polylogarithmic range between $\Omega\left(\sqrt{\log n}\right)$ and $O(\log^2 n)$. This is an \emph{exponential improvement} in the previous upper bound of $O(\sqrt{n \log n})$ for self-destructive competition shown by Cho et al.~\cite{cho2021distributed} in the special case with $\delta=0$ (i.e., a model with no individual death reactions).
Our result applies to a much larger class of models: for any choice of $\beta,\delta,\alpha>0$, the chain reaches majority consensus with high probability provided that the gap is $\Omega(\log^2 n)$. In contrast, we show that if the gap is $o(\sqrt{\log n})$ the chain fails to reach majority consensus with constant probability. 

Second, for non-self-destructive interference competition, we show that the majority threshold lies in a polynomial range between $\Omega(\sqrt{n})$ and $O(\sqrt{n \log n})$. This shows an \emph{exponential separation} between competitive Lotka--Volterra models with self-destructive and non-self-destructive interspecific interference competition in these models.

In comparison, recently Andaur et al.~\cite{andaur2021reaching} gave an upper bound of $O(\sqrt{n \log n})$ for a slightly different model with resource-consumer dynamics and interference competition. However, their proof is only for models without individual death reactions ($\delta=0$), assumes bounded non-mass-action growth dynamics, and their analysis only guarantees majority consensus with probability $1-O(1/\sqrt{n})$ and not with true high probability (i.e., probability $1-1/n^c$ for any constant $c>0$).

Our upper bound holds for any constant $\alpha>0$ and $\beta,\delta \ge 0$, i.e., our model also allows for individual death reactions. Moreover, we can show that majority consensus is reached with high probability provided that the gap is $\Omega(\sqrt{n \log n})$. Finally, proof technique can also be applied to the model of Andaur et al.~\cite{andaur2021reaching} in a straightforward manner, yielding a stronger probability guarantee for majority consensus also in their model with bounded, non-mass-action growth rates.

  For both self-destructive and non-destructive interspecific competition, we show that consensus is reached within $O(n)$ events both in expectation and with high probability in the absence of intraspecific competition.

\paragraph{Intraspecific competition.}
Finally, unlike prior work, we also investigate majority consensus in systems with \emph{intraspecific interference competition}, i.e., competition between the individuals of the same species. We show that systems with intraspecific competition can be have fundamentally different behavior in terms of majority consensus. Namely, we show that such majority consensus thresholds do not always exist. The results are given in \appendixref{sec:intra-bounds}.

    First, if intraspecific and interspecific competition are equally strong ($\alpha=\gamma$) for self-destructive competition, then the probability of majority consensus is equal to the initial proportion of the majority species. This implies that majority consensus threshold is at $n-1$, and that such systems cannot solve majority consensus with high probability (in the true sense). A similar result holds for non-self-destructive competition with $\gamma=2\alpha$. Second, we also show that systems with only intraspecific interference competition have no majority consensus thresholds: given any gap, the chain fails to reach majority consensus with at least positive constant probability.

    \subsection{Technical challenges}

For the stochastic LV models without intraspecific competition, our analysis uses a new  technical approach for bounding the noise arising from individual reactions and asymmetric outcomes of non-self-destructive competition. We introduce a new ``asynchronous, pseudo-coupling'' technique that can be used to bound the behaviour of the two-species process before it reaches consensus using what we call ``nice'' single-species birth-death chains.

On a high-level, this approach is similar in spirit to the recent coupling techniques of Andaur et al.~\cite{andaur2021reaching} and Cho et al.~\cite{cho2021distributed}, which couple a two-species chain with easier to analyze single-species birth-death chains. However, there are key differences, and our approach is more general.

\paragraph{Limitations of prior coupling approaches.}
The existing techniques for microbial majority consensus~\cite{andaur2021reaching,cho2021distributed}
are closely tailored to specific model variants they consider. When trying to use these techniques to the LV models,  both of the previously existing coupling techniques critically break when the process is allowed to have certain stochastic events that may occasionally benefit the minority species (e.g., individual death reactions of the majority species decreasing the discrepancy).

Cho et al.~\cite{cho2021distributed} use elaborate coupling arguments between several different chains to first establish a bound on the extinction time, and then couple the two-species chain to parallel birth-only Yule processes to bound the probability of reaching majority consensus using the regularized incomplete beta function. In addition to the actual coupling argument, the proof critically assumes that no individual deaths can occur and that competition is self-destructive. Moreover, the bound obtained with this technique is exponentially far from the real bound, as we show. The submartingale argument of Andaur et al.~\cite{andaur2021reaching} allows for non-self-destructive competition at the cost of forbidding individual death reactions and losing the guarantee of true with high probability majority consensus.

\paragraph{Overview of new techniques.}
We resolve the limitations of prior techniques by moving from usual Markovian couplings that update the single-species chain and the two-species chain simultaneously to a new type of an ``asynchronous pseudo-coupling''. That is, our construction is not strictly speaking a coupling between the two-species and single-species processes, as the coupled chains are not updated in lock-step. Despite this complication, we can still carefully extract useful information about the distribution of events in the two-species chain using stopping time arguments.

This is achieved in part by also deriving a much more detailed accounting of the (demographic) noise arising from birth reactions of the minority species and death reactions of the majority species. We abstract the properties of the single-species chain required by the pseudo-coupling, and providing a more fine-grained analysis of the distribution of events in the single-species chain.

Furthermore, our technique makes far fewer assumptions about the structure of the two-species system and properties of the dominating chains. Therefore, we suspect that our approach can be further extended to analyze more realistic models beyond Lotka--Volterra dynamics such as models incorporating general non-bounded resource-consumer dynamics (i.e., exploitative competition) in addition to interference competition. Indeed, our new techniques give both a simpler analysis and stronger guarantees of the previously studied microbial majority consensus dynamics~\cite{andaur2021reaching,cho2021distributed}.

\paragraph{Our approach in a nutshell.}
Our main conceptual approach is to analyze the probability of reaching majority consensus $\rho(\vec S)$ by considering the following random-length sum defined by
\begin{equation}\label{eq:noise-sum}
F(\vec S) = \sum_{t=1}^{T(\vec S)} F_t,
\end{equation}
where $F_t = \Delta_{t-1} - \Delta_{t}$. Note that $F_t > 0$ if the discrepancy changed in favor of the initial minority species and $F_t < 0$ if the discrepancy changed in favor of the initial majority species at time step $t$.

The random variable $F(\vec S)$ counts how much the initial gap changes in favor of the initial \emph{minority} species. This captures the effects of \emph{demographic noise} on the difference between the initial majority and minority species before consensus is reached. Provided that $T(\vec S)$ is finite with probability 1, we have that
$\rho(\vec S) = \Pr[F < \Delta_0] = 1 - \Pr[F \ge \Delta_0]$,
that is, we can connect the probability of reaching majority consensus with an given initial gap to the cumulative density function of $F$. Since the random variables $T(\vec S)$ and $F_1, \ldots, F_{T(\vec S)}$ are dependent, getting a handle on the distribution of $F$ can be non-trivial. Moreover, in general $(F_t)_{t \ge 0}$ does not give a (sub)martingale, so it is not clear if one can easily employ readily available martingale concentration bounds to analyze the process.

To bound $F$, we observe that the demographic noise $F = F_\textrm{ind} + F_\textrm{comp}$ can be divided to two components, where $F_\textrm{ind}$ is the noise arising out of reproductive dynamics (i.e., individual birth and death reactions) and $F_\textrm{comp}$ is the noise arising from competition dynamics. This allows us to obtain a much more refined bound on total noise in the system by studying both components separately, which previous techniques were unable to do.

Intuitively, under self-destructive competition, there is only noise from the individual events, which are fairly rare under the mass action dynamics. This noise turns out to be ``polylogarithmic''. On the other hand, under non-self-destructive competition, there is additional noise coming from the chance outcomes of the competition events. We show there are $O(n)$ competition events before extinction. The outcomes of the competition events are (intuitively) similar to a random walk on the line. The key challenge is that the events causing different types of noise are interleaved. 

\subsection{Open problems and future directions}

We suspect that in general the bound of $O(\log^2 n)$ for self-destructive competition is not tight for all parameter ranges (e.g., it clearly is not tight in the case $\beta=\delta=0$). The natural step is to identify the tight asymptotic majority consensus threshold in the case $\alpha,\beta,\delta>0$. Second, while we focus on Lotka--Volterra models, we believe our techniques are applicable to a wider variety of stochastic population models beyond the competitive, two-species Lotka--Volterra model.

We show that in the worst-case, intraspecific competition can badly hinder the probability of reaching majority consensus. For example, with self-destructive competition, we identify that the majority consensus threshold is $O(\log^2n)$ with $\alpha>0$ and $\gamma=0$ and $n-1$ when $\alpha=\gamma>0$. An interesting open problem is to identify at which point does the majority consensus threshold enter a sublinear or polylogarithmic regime when $\alpha>0$ is a fixed constant and $\gamma \to 0$.

On a conceptual level, our results suggest interesting \emph{computational trade-offs} in the design of microbial circuits: majority consensus protocols utilizing self-destructive competition seem to be less sensitive to demographic noise than protocols based on non-self-destructive competition. However, the trade-off is that the former is much more costly at the individual level. From a bioengineering perspective, it would be interesting to investigate if circuits utilizing non-self-destructive interference are evolutionarily more stable than circuits utilizing self-destructive competition.

Finally, one may surmise that the computational trade-offs implied by this work are solely theoretical. For example, idealized well-mixed, mass action Lotka--Volterra models do not capture the full range of microbial dynamics. However, such models have been experimentally  observed to provide reasonable approximations in many situations~\cite{gonze2018microbial,hu2022emergent,dedrick2023does} also in the context of synthetic microbial consortia~\cite{alnahhas2020majority,li2022synthetic}. Regardless, future work should further investigate how relaxing the model assumptions influence the predicted computational trade-offs and experimentally test this. This requires
developing new proof techniques for dealing with non-mass action models and/or explicit spatial dynamics.

\subsection{Limitations and biological assumptions}\label{sec:limitations}

This paper focuses on new mathematical techniques for the analysis of distributed consensus dynamics in stochastic interaction models. Like in all biological models, there are several biological assumptions and restrictions made in the stochastic Lotka--Volterra models we study. Indeed, our work shows how subtle differences in model assumptions can fundamentally impact the solvability and complexity of majority consensus in microbial population models.

\paragraph{Markovian, mass action dynamics.}
As typical for analytical models, we assume memoryless species and mass action dynamics.
We do not model the internal biophysical or metabolic state of the individual bacteria.
In reality, the reproductive dynamics of microbes can
depend on their internal state and changing environmental conditions~\cite{gonze2018microbial}.
While such models may very well be used in numerical simulations, their analysis seems to be far beyond the reach of current proof techniques.
Regardless, our technique can still be used to analyze also models with non-mass-action reactions, such as those in the biological reaction network model of Andaur et al.~\cite{andaur2021reaching}.

\paragraph*{Physical aspects and asymptotics.}
We assume unit volume.
This differs from how molecular systems are analyzed at the so-called thermodynamic limit, where the ratio of volume and initial population size remains constant as $n \to \infty$.
We do not assume such a constant initial population density.

While the limit of an infinitely dense cell population is clearly non-physical, the asymptotic analysis here illuminates the impact of different biological mechanisms
on the performance of microbial protocols
with
increasing initial cell populations (i.e., larger inputs).
In a wetlab setting, the initial population size can typically vary over a large range, from $1\,\text{mL}^{-1}$ to about $10^9\,\text{mL}^{-1}$ for \emph{E. coli}. Nevertheless, it remains open to demonstrate that within this range, asymptotic effects are dominating, as suggested by numerical simulations done in prior work (see, e.g., \cite{cho2021distributed}).

Finally, it is important to note that while our LV models allow for increasing initial density and there is no explicitly assumption bounding the maximum size of a population, the population size is regulated by a carrying capacity induced by density-dependent competition. In particular, when including also intraspecific competition ($\gamma>0$), the stochastic LV models exhibits the full logistic growth regime usually observed for microbial populations even after competitive exclusion.

\subsection{Structure of the paper}
We have delegated some of the proofs to the Appendix. In addition, \sectionref{sec:related} further discusses the background on previous theoretical and empirical work regarding Lotka--Volterra models and related work on majority consensus in other models of computing. We start with preliminaries in \sectionref{sec:preliminaries}.
In \sectionref{sec:nice}, we establish some useful results regarding the behavior of certain single-species chains, which we then use to bound the behavior of competitive two-species Lotka--Volterra chains in \sectionref{sec:dominating-chains}. \sectionref{sec:self-destructive} and \sectionref{sec:non-self-destructive} give bounds for majority consensus thresholds under self-destructive and non-self-destructive interference competition.
The lower bounds for majority consensus thresholds in systems with intraspecific competition appear in \sectionref{sec:intra-bounds}.

\section{Related work}\label{sec:related}

\subsection{Lotka-Volterra models}

Named after the pioneering work of Lotka~\cite{lotka1925elements} and Volterra~\cite{volterra1926fluctuations}, Lotka--Volterra (LV) models have been the cornerstone of theoretical biology and microbial ecology for over one hundred years. Gause~\cite{gause1934} reported the first experimental validations of the two-species Lotka--Volterra equation using yeast and protozoa microcosms in 1934. Since then, LV models have become widely applied in microbial ecology~\cite{gonze2018microbial}. They have been found to predict well even many aspects of complex microbial communities~\cite{dedrick2023does,hu2022emergent,stein2013ecological}. Given the abundance of work on LV models, we only give a brief overview, focusing on competitive LV models in well-mixed populations.

In biology, the competitive Lotka--Volterra ordinary differential equation model has become the baseline model for ecological community dynamics~\cite{hofbauer1998evolutionary,murray2002mathematical,may2007theoretical}.
The simplicity of numerical simulation of ODE models has rendered them a popular choice for modelling synthetic microbial consortia dynamics~\cite{li2022synthetic,mao2015slow,zomorrodi2016synthetic}.
These deterministic models are derived under the assumption of (infinitely) large continuous populations, ignoring the important stochastic effects driving the dynamics of real (finite and discrete) populations~\cite{lande2003stochastic,black2012stochastic,zhou2017stochastic}.
Despite their shortcomings, the key reason for the popularity of the deterministic models is simple: even complex ODE models are easy to simulate numerically, while even simple stochastic models tend to be difficult to analyse~\cite{jagers2010plea,geritz2012mathematical,black2012stochastic}.

This challenge has motivated a significant line of research devoted to the analysis of stochastic, discrete Lotka--Volterra models in well-mixed populations, typically using the tools of statistical mechanics~\cite{mckane2004stochastic,huang2015stochastic,reichenbach2006coexistence,dobrinevski2012extinction}. These models closely resemble to the models we use, and many papers in this area also investigate coexistence and extinction probabilities (as a function of time). However, they do not investigate the probability of \emph{majority consensus} as a function of $n$ and $\Delta$, nor compare how different choices of competitive mechanisms impact these probabilities. Furthermore, these models usually assume only non-self-destructive competition and finite maximum population size. This subtly differs from our model, where the carrying capacity arises from interference competition. In the theoretical computer science literature, convergence and thresholds of discrete LV dynamics have been studied in the population protocol model~\cite{czyzowicz2022convergence} in bounded populations.

\paragraph{Comparison with deterministic kinetics.}
We note that for the two stochastic Lotka--Volterra models we study,
the continuous-valued approximations
under deterministic mass action kinetics correspond to the usual deterministic competitive Lotka--Volterra equations~\cite{gonze2018microbial,hofbauer1998evolutionary}. For two-species in the neutral case, this is given by  the non-linear differential equations  \begin{equation}\label{eq:two-ode}
\frac{\operatorname{d}x_i}{\operatorname{dt}} = x_i(r - \alpha' x_{1-i} - \gamma' x_i),
\end{equation}
where $x_i$ is the density of species $i \in \{0,1\}$,  $r=\beta-\delta$ is the intrinsic growth rate, $\gamma'=\gamma_0=\gamma_1 \ge 0$ is the rate of intraspecific competition, and $\alpha' \ge 0$ is the rate of interspecific competition so that for model (1), we have $\alpha' = \alpha = \alpha_0+\alpha_1$, and for model (2) we have $\alpha'=\alpha_0 = \alpha_1$.
It is easy to see that in this  model, if $\alpha' > \gamma'$, then the species with the higher initial density will deterministically always win. Thus, this model fails to capture the stochastic effects occuring in finite populations.

\subsection{Majority consensus in stochastic interaction models}

The majority consensus problem has been studied in many asynchronous, stochastic interaction models in fully-connected networks (i.e., well-mixed systems). The  problem and its extension, plurality consensus, have also been studied in synchronous, gossip models with static population size~\cite{cooper2014power,becchetti2014plurality,feinerman2017breathe,ghaffari2016polylogarithmic,bankhamer2022fast}. However, here we primarily focus on stochastic asynchronous models that most closely resemble our setting.

\paragraph{Majority consensus in stochastic biological population models.}
The closest to our work is the recent work by Cho et al.~\cite{cho2021distributed} who considered a special case of our discrete, stochastic Lotka--Volterra model in a two-species chemical reaction network model. They showed that an initial gap of $\Omega(\sqrt{n \log n})$ is sufficient under the assumption $\delta=0$ and self-destructive interspecific interference competition. In contrast to their results, our results hold also for systems with $\delta>0$ and show that an exponentially smaller gap of only $\Omega(\log^2 n)$ suffices for any constant rate constants $\alpha_0, \alpha_1, \beta,\delta > 0$.

More recently, Andaur et al.~\cite{andaur2021reaching} considered a resource-consumer model of biological population dynamics with non-self-destructive interference competition. Their model is not strictly speaking a chemical reaction network model, as they allow for non-mass-action reactions. They showed that in their model $\Omega(\sqrt{n \log n})$ gap suffices with probability $1-O(1\sqrt{n})$.
However, their model assumes bounded, non-mass-action growth and no individual death reactions (i.e., $\delta = 0$). With some work, our new technique can also applied to their model to obtain a bound of $\Omega(\sqrt{n \log n})$ with high probability, i.e., with probability at least $1-1/n^k$, for any constant $k>0$. This is because their dominating chain is also a ``nice chain'' in the sense we define in \sectionref{sec:nice}.

\paragraph{Chemical reaction networks.}
Condon et al.~\cite{condon2020approximate} consider the majority consensus problem and multi-valued consensus in general chemical reaction systems that also allow \emph{trimolecular} reactions (i.e., a single reaction can involve up to three reactants and products). They gave bi- and trimolecular reaction networks for which an initial gap of $\Omega(\sqrt{n \log n})$ is sufficient with high probability.

While this work operates in the same formal model, the protocols are not directly comparable with protocols in the two-species Lotka--Volterra model.  However, interestingly, their ``heavy-B'' protocol, in which reactions have two reactants and three products, and bimolecular ``double-B'' protocol with two reactants and two products for each rule, resemble our protocols with self-destructive competition. On the other hand, the rules in the protocol ``single-B'' resemble non-self-destructive competition. However, these protocols employ three species. They also give a two-species trimolecular protocol that uses three reactants and products.

\paragraph{Population protocols.}
The population protocol model~\cite{angluin2006computation} has become a popular model to study computation in well-mixed chemical solutions. The model is a special case of the chemical reaction network model, where each rule has exactly two reactants and products, and each reaction has unit rate. In particular, this implies that the total population size $n$ remains static throughout the execution of the protocol.

For majority consensus, Angluin et al.~\cite{angluin2008simple} considered approximate majority agreement in the population protocol model. They give a simple 3-state algorithm that succeeds with high probability when the initial gap is $\Omega({\sqrt{n}} \cdot \log n)$ and the number of steps to reach consensus is $O(n \log n)$ with high probability. The same cancellation-idea used by this protocol also appears in protocols in a variety of other models, including our Lotka--Volterra protocols and other protocols relying on competition~\cite{andaur2021reaching,cho2021distributed}, the
``single-B'' protocol of Condon et al.~\cite{condon2020approximate} and the ternary signalling protocol of Perron et al.~\cite{perron2009using}.  Perron et al.~\cite{perron2009using} analysed such a cancellation protocol and showed that if the gap is linear gap the protocol succeeds with very high probability, i.e., fails to converge to the initial majority value with exponentially small probability.

Draief and Vojnovi\'{c}~\cite{draief2012-convergence} showed that there exists a 4-state protocol that always reaches majority consensus  with \emph{any} positive gap in $O(n^2)$ expected interactions. The same protocol was also described and analysed by Mertzios et al.~\cite{mertzios2014majority}.
This variant of majority consensus, where the protocols are required to succeed with probability 1, is often referred to \emph{exact majority} in the population protocol literature.
Alistarh et al.~\cite{alistarh2015-fast} showed that using $O(\log^3 n)$ states, exact majority can be solved in only $n \polylog n$ interactions in expectation and with high probability.

Subsequently, Alistarh et al.~\cite{alistarh2017time} showed that any $o(\log \log n)$-state protocol for exact majority that succeeds requires $n^2/\polylog n$ interactions. They also showed that using $O(\log^2 n)$ states majority can be solved in $O(n \log^2 n)$ interactions.
The upper bound result was subsequently improved to $O(\log^2n)$ states and the lower bound strengthened to $\Omega(\log n)$ for a large class of protocols by Alistarh et al.~\cite{alistarh2018space-optimal}. Recently,
Doty et al.~\cite{doty2022time} gave a $O(\log n)$-state protocol that stabilises in $\Theta(n \log n)$ expected interactions, matching the lower bounds.

Czyzowicz et al.~\cite{czyzowicz2022convergence} considered population protocols with discrete Lotka--Volterra -like stochastic dynamics. Their Lotka--Volterra dynamics are different, as the operate in the population protocol model, where the total population size remains static. In particular, there are no individual birth and death reactions.  In this setting, they gave a 4-state protocol that solves majority consensus with high probability provided that the ratio of the initial input counts $a$ and $b$ is $a/b = 1+\varepsilon/(1-\varepsilon)$ for a constant $\varepsilon>0$, that is, the initial gap is linear.  For their analysis, they also use a coupling technique with delays, but their technique is different, as they operate in the population protocol model with a static maximum size for the system.

While the population protocol model is in spirit similar to the stochastic biological population models we consider, the full range of algorithmic techniques available in the population protocol model are not easily (or if at all) realizable in synthetic biology setting. For example, one cannot completely enforce which reactions are active (e.g., by assuming that no birth/death events occur or that they occur only during certain phases of the protocol), a technique used by cancellation-doubling protocols for fast exact majority protocols. In particular a key challenge in the synthetic biology setting is that reproductive and ecological dynamics are interleaved with the engineered biological circuits, which in combination give rise to the microbial protocol~\cite{alnahhas2020majority,li2022synthetic}.

Finally, we note that Goldwasser et al.~\cite{goldwasser2018population} studied the impact of changing population composition in a synchronous variant of the population protocol model, where an adversary is allowed to make a bounded number of \emph{arbitrary} insertions and deletions of individuals per round. They showed that even under such adversarial changes to population structure, a certain population balancing problem can be solved efficiently with only small number of states. In their setting, the individuals can also trigger duplication (birth) and self-destruction (death) events. In contrast, in our setting, the birth and death events are dictated by the underlying stochastic Lotka--Volterra dynamics.

\section{Preliminaries}\label{sec:preliminaries}

We use $\mathbb{N} = \{0,1, \ldots \}$ to denote the set of nonnegative integers.
For any $n > 0$, we use $\log n$ for the base-2 logarithm of $n$ and $\ln n$ for the natural logarithm of $n$. For any integer $n>0$, we write $H_n$ to denote the $n$\textsuperscript{th} Harmonic number given by the sum $\sum_{i=1}^n 1/i$, which is lower-bounded by $\ln n$.

\paragraph{Concentration bounds.}
We say that a random variable $X$ is a \emph{Bernoulli} random variable if $X$ only takes values in $\{0,1\}$. A random variable $X$ is said to be $O(f(n))$ \emph{with high probability} if for any fixed constant $k \ge 0$ there exists a constant $C(k)$ such that $\Pr[X \ge C(k) f(n)] \le 1/n^k$.
To establish with high probability bounds, we use two standard results on the concentration of random sums.

\begin{lemma}[Chernoff bounds]\label{lemma:chernoff}
Let $X = X_1 + \cdots + X_n$ be the sum of $n$ independent  Bernoulli random variables. Then
\begin{enumerate}\item $\Pr[X \ge (1+\varepsilon) \cdot \E[X]] \le \exp\left(- \E[X] \cdot \varepsilon^2
    / (2+\varepsilon) \right)$ for any $\varepsilon > 0$, and
    \item $\Pr[X \le (1-\varepsilon) \cdot \E[X] \le \exp\left(- \E[X] \cdot \varepsilon^2/2\right)$ for any $0 < \varepsilon < 1$.
\end{enumerate}
\end{lemma}

We will also make use of a special case of Hoeffding's inequality for random variables restricted to the range $[-1,1]$.

\begin{lemma}[Hoeffding's inequality]\label{lemma:hoeffding}
Let $X = X_1 + \cdots + X_n$ be the sum of $n$ independent  random variables, where $X_i \in [-1,1]$. Then for any $t \ge 0$
\[
\Pr\left[ \left| X - \E[X]\right| \ge t  \right] \le 2\cdot \exp\left(- \frac{2t^2}{n}\right).
\]
\end{lemma}

For our lower bounds, we also make use of the following \emph{anti-concentration bound} implied by the Central Limit Theorem. 

\begin{lemma}\label{lemma:clt}
  Let $\varepsilon \in (0,1)$ be a constant and $X$ be the sum of $n$ i.i.d.\ random variables $X_1, \ldots, X_n$ with mean 0 and variance 1.  Then there is a constant $\theta = \theta(\varepsilon) > 0$ such that for any sufficiently large $n$, we have $\Pr[X > \theta \sqrt{n}] \ge \varepsilon$.
\end{lemma}

\paragraph{Stochastic domination and couplings.}
Let $X_1$ and $X_2$ be two non-negative random variables. If $\Pr[X_2 \ge x] \ge \Pr[X_1 \ge x]$ for all $x \ge 0$,
then we write $X_1 \preceq X_2$ and say that $X_2$ \emph{stochastically dominates} $X_1$. The random variable $(\widehat{X_1}, \widehat{X_2})$ is said to be a \emph{coupling} of $X_1$ and $X_2$ if the distribution of $\widehat{X_i}$ is the same as the distribution $X_i$, that is, $\Pr[X_i \ge x] = \Pr[\widehat{X}_i \ge x]$ for any $x \ge 0$. We make use of the following simple lemma.

\begin{lemma}\label{lemma:couple-with-independent}
Let $X = X_1 + \cdots + X_n$ be a sum of (not necessarily independent) Bernoulli random variables and $Y = Y_1 + \cdots + Y_n$ be a sum of \emph{independent} Bernoulli random variables. \begin{enumerate}\item If $\Pr[X_i = 1 \mid X_1, \ldots, X_{i-1}] \le \Pr[Y_i=1]$ for each $1 \le i \le n$, then $X \preceq Y$.

    \item If   $\Pr[X_i = 1 \mid X_1, \ldots, X_{i-1}] \ge \Pr[Y_i=1]$ for each $1 \le i \le n$, then $Y \preceq X$.
\end{enumerate}
\end{lemma}

The proofs of \lemmaref{lemma:clt} and \lemmaref{lemma:couple-with-independent} are given in \appendixref{apx:omitted}.

\section{Bounds for nice single-species chains}\label{sec:nice}

Let $p,q \colon \N \to [0,1]$ such that $p(n) + q(n) \le 1$ for all $n\geq0$. The birth-death chain defined by $p$ and $q$ is the discrete-time Markov chain $N = \left( N_t \right)_{t \in \N}$ on the state space $\mathbb{N}$, where in each step the chain goes from state $n$ to $n+1$ with \emph{birth probability} $p(n)$, to state $n-1$ with \emph{death probability} $q(n)$. The probability $h(n)=1-p(n)-q(n)$ is the \emph{holding probability} of the chain in state $n$. A state~$x$ is \emph{absorbing} if $p(x)=q(x)=0$.

We assume that $p(n)>0$ and $q(n)>0$ for all $n >0$ and $p(0)=q(0)=0$ so that 0 is the unique absorbing state. The \emph{absorption time} or \emph{extinction time} of a chain $N$ is $E(N) = \min \{ t : N_t = 0 \}$,
that is, the minimum time until the chain reaches the unique absorbing state. We say the chain is \emph{nice} if there exist constants $C,D > 0$ such that $p(n) \le C/n$ and $q(n) \ge D$ for all $n>0$.

\paragraph{Bounds on number of births for nice chains.}
For any nice chain started, we will bound the extinction time $E(N)$ of a chain and the number $B(N)$ of birth events that occur before extinction. We assume throughout that $N_0 = n$, and with a slight abuse of notation, we write $E(N)=E(n)$ and $B(N)=B(n)$ for such a chain $N$.
We first asymptotically bound the extinction time. The lower bound follows immediately from the fact that the chain needs to decrement at least $n$ times to reach state 0.
For the upper bound, the result follows from a known result for the absorption time of discrete-time birth-death processes; see \cite[Theorem 3.1]{sericola2013birth} and \cite[Lemma 3]{andaur2021reaching}.
\begin{lemma}
  For any nice chain started in state $n$, its expected extinction time is $\E[E(n)] = \Theta(n)$.\label{lemma:nice-extinction}
\end{lemma}

Equipped with the above lemma, we can bound the number of birth events before extinction.

\begin{lemma}
  For any nice chain, the expected number of births satisfies $\E[B(n)] \in O( \log n)$. \label{lemma:nice-expected-births}
\end{lemma}
\begin{proof}
By \lemmaref{lemma:nice-extinction} and the law of total expectation, for any non-negative integer-valued random variable $X$ we have the bound
\begin{align*}
\E[E(X)] &= \sum_{k=0}^\infty \Pr[X = k] \cdot \E[E(k)] \le C' \sum_{k=0}^\infty \Pr[X = k] \cdot k = C' \E[X],
\end{align*}
where $C'$ is the constant of the upper bound from \lemmaref{lemma:nice-extinction}.

We say that a time step $t$ is a holding step if the chain stays in the same state (i.e., neither a birth or death happen).
Let $R = R(n)$ be the minimum time until exactly $n$ non-holding steps have occurred after starting the chain in state $n$. Clearly, $n \le R(n) \le E(n)$, since the chain needs to decrement at least $n$ times before going extinct. Moreover, $R(n)$ is finite with probability 1, as $E(n)$ is finite with probability 1.
Let $B_t$ denote the number of birth events and $D_t$ denote the number of death events that have occurred by step $t \ge 0$.
Observe that $D_R + B_R = n$, as at time step $R$, there have been exactly $n$ non-holding steps. Thus, at time $R$ the chain is in state
\[
N_R  = n - D_R + B_R = n - (n-B_R) + B_R = 2B_R.
\]
Moreover, note that
\[
B(n) \le B_R + E(N_R) = B_R + E(2B_R),
\]
since the number of birth events after step $R$ is upper bounded by the extinction time of the chain started in state $N_R$. Thus, by linearity of expectation and the above, we get that
\begin{align*}
\E[B(n)] &\le \E[B_R] + \E[E(2B_R)] = (2C'+1) \cdot \E[B_R].
\end{align*}
We next show that $\E[B_R] \in O(\log n)$, which implies the claim of the lemma.

Let us consider the value of $B_R$. Since $R \ge n$ is the minimum time until exactly $n$ non-holding steps have happened, we have
\[
B_R = \sum_{i=1}^n X_i,
\]
where $X_i$ is the indicator variable whether the $i$th non-holding step is a birth event or not. Note that $\E[X_i] \le C/(n-i+1)$, as the chain is in state at least $n-i+1$ after $i$ non-holding steps.
Therefore, the expectation of $B_R$ satisfies
\[
\E[B_R] = \E[X] = \E\left[ \sum_{i=1}^n X_i \right] \le C \sum_{i=1}^n \frac{1}{n-i+1} = C \sum_{i=1}^n \frac{1}{i} = CH_n,
\]
where $H_n$ is the $n$th Harmonic number. This proves the claim.
\end{proof}

With the bound on expected number of births, some calculations and the application of Markov's inequality and Chernoff bounds yield the following claims for all sufficiently large $n$.

\begin{lemma}\label{lemma:nice-whp-births}
  Let $k>0$. For any nice chain, there exists a $C(k)$ such that $\Pr[B(n) \ge C(k) \cdot \log^2 n] \le 1/n^k$.
\end{lemma}
\begin{proof}
  Let $c>0$ such that $\E[B(n)] \le c \log n$ for any $n>0$; by \lemmaref{lemma:nice-expected-births} such a constant exists.
  Set $L = \lceil 2c\log n \rceil$, $K = \lceil k \log n \rceil$, and $M = n + LK$. We will assume that $n$ is sufficiently large so that $M \le n^2$ holds. We will make $K$ trials, where in each trial $1 \le i \le K$, we run the chain until either $L$ birth events occur or the chain hits state 0. After the $i$th trial, the chain is in some state $n_i \le n+iL \le n + kL \le M \le n^2$. We say that the $i$th trial is successful if the chain hits the state 0, and otherwise, it fails. Note that $B(n_i) \preceq B(n^2)$ for each $1 \le i \le K$ and that $\E[B(n^2)] \le c\log n^2 = 2c \log n \le L$.
Since the chain is Markovian and by Markov's inequality, the probability that the $i$th experiment fails is at most
\[
\Pr[B(n_k) \ge L] \le \Pr[B(n^2) \ge L] \le \Pr[B(n^2) \ge 2 \cdot \E[B(n^2)]] \le \frac{1}{2}.
\]
The probability that all $K$ experiments fail is thus at most $1/2^K \le 1/n^k$. After the $K$ experiments, there have been at most $KL \in O(\log^2 n)$ birth events. Hence the claim follows.
\end{proof}

\begin{lemma} \label{lemma:nice-whp-extinction}
  Let $k > 0$. For any nice chain, there exists a $\theta(k)$ such that
$\Pr[E(n) \ge \theta(k)\cdot n] \le 1/n^k$.
\end{lemma}
\begin{proof}
  Since the chain is nice, there exists a constant $D \in (0,1)$ such that  $q(m) \ge D$ for all $m>0$. Set $R = \lceil 6n/D \rceil$.
  We first bound the number $X$ of holding steps in non-absorbing states (i.e., states other than 0) that occur during the first $R$ steps. Let $X_t$ be the indicator variable denoting whether a holding step in a non-absorbing state occurs at step $t \ge 1$.
Clearly, $N_t = 0 \implies X_t = 0$. For any $n>0$,
the probability of such a holding step in state $n$ is at most $1-p(n)-q(n) \le 1-D$, i.e., $\Pr[X_t=1 \mid X_1, \ldots, X_{t-1}] \le 1 - D$.
  Let $Y_1, \ldots, Y_R$ be a sequence of independent Bernoulli random variables with $\Pr[Y_t = 1 ] = 1-D$. By \lemmaref{lemma:couple-with-independent}, we have that
  \[
X = \sum_{t=1}^R X_t \preceq \sum_{t=1}^R Y_t = Y.
  \]
Note that $\E[Y] = (1-D)R \in \Theta(n)$ and set $\varepsilon = \sqrt{3(k+1) \ln n / \E[Y]}$. Since $\E[Y] \in \Theta(n)$, we have $\varepsilon \in o(1)$.
Applying the Chernoff bound from \lemmaref{lemma:chernoff} with $\varepsilon$ yields
\[
\Pr[X \ge (1+\varepsilon) \E[Y]] \le \Pr[Y \ge (1+\varepsilon)\E[Y]] \le \exp\left( - \frac{\varepsilon^2}{3} \cdot \E[Y] \right ) \le 1/n^{k+1}.
  \]
Moreover, by \lemmaref{lemma:nice-whp-births} we get that with probability at most $1/n^{k+1}$ there are more than $C(k+1) \log^2 n$ birth events before the chain goes extinct. Thus for sufficiently large $n$, with probability $1-2/n^k$, by time $R$ there are at most
\[
K = (1+\varepsilon) \E[Y] + C(k) \log^2 n \]
birth events or holding steps in non-absorbing states.
For large enough $n$, we have that $\varepsilon < D/2$ and $C(k) \log^2 n \le n$.
Now
\begin{align*}
  R - K &= R - (1+\varepsilon) \E[Y] - C(k) \log^2 n \\
   &\ge R - \left(1+\frac{D}{2}\right) \left(1-D\right) R - n
  \ge R - \left(1-\frac{D}{2} - \frac{D^2}{2}\right) R - n \\
  &= R \left(\frac{D}{2} + \frac{D^2}{2} \right ) - n
  >\frac{DR}{2} -n \ge \frac{6n}{2} - n = 2n.
\end{align*}
Therefore, out of the $R$ steps, there are at least $R-K > 2n$ steps that are either holding steps in state 0 or death events. As the chain never reaches a state higher than $n + O(\log^2n)$ with high probability, we get that the chain hits the state 0  in $R$ steps with probability at least $1-2/n^k$ for all sufficiently large $n$, as $n + C(k+1)\log^2n < 2n$.
\end{proof}

\section{Dominating chains for Lotka--Volterra systems}\label{sec:dominating-chains}

We now introduce a ``pseudo-coupling'' of single-species birth death chains and two-species Lotka--Volterra chains that can be used to \emph{over-approximate} the consensus time and the number of ``bad'' individual events that can decrease the gap. Unlike previous ``dominating chain'' approaches~\cite{andaur2021reaching,cho2021distributed}, which consider restricted choices of rate parameters, our approach works with any choice of
$\beta,\delta,\alpha_0,\alpha_1>0$.
In particular, we allow individual death reactions with $\delta>0$ and do not assume symmetric interference competition, i.e., $\alpha_0 \neq \alpha_1$ is possible. The condition for the dominating chain is also simpler; this comes with the cost that we do not get a coupling in the strict sense.

\subsection{The chain domination lemma}

Let $\vec S = (\vec S_t)_{t \ge 0}$ be a two-species chain. We say that an event at time step $t$ is \emph{bad} if it decreases the gap between the minimum and maximum species, i.e., if $\Delta_{t+1} = \Delta_t - 1$ holds \emph{conditioned on} that the minimum species has positive count. Let $P(a,b)$ be the probability that the chain $\vec S$ in state $(a,b)$ has a bad \emph{non-competitive} reaction. We say that an event at time step $t$ is \emph{good} if the species with the smaller count decreases in count. Let $Q(a,b)$ be the probability that the chain $\vec S$ in state $(a,b)$ has a \emph{good} reaction.
Note that the probability of bad competitive events is $1-P(a,b)-Q(a,b)$.

 Let $N = (N_t)_{t \ge 0}$ be a single-species birth-death chain defined by birth function $p$ and death function $q$. We say that $N$ is a \emph{dominating chain} for $\vec S$ if for any $a,b \ge 0$ we have
  \begin{enumerate}[noitemsep]
    \item[(D1)] $P(a,b) \le p(\min\{a,b\})$, and
    \item[(D2)] $Q(a,b) \ge q(\min\{a,b\})$.
    \end{enumerate}

\begin{lemma}[Chain domination lemma]\label{lemma:chain-domination}
  Suppose the single-species chain $N = (N_t)_{t \ge 0}$ is a dominating chain for the two-species chain $\vec S = (\vec S_t)_{t \ge 0}$. If $N_0 \ge \min \vec S_0$, then
\begin{enumerate}[noitemsep,label=(\alph*)]
\item $T(\vec S) \preceq E(N)$, and
\item $J(\vec S) \preceq B(N)$,
\end{enumerate}
 where $T(\vec S)$ is the consensus time of $\vec S$, $J(\vec S)$ the number of bad non-competitive reactions in the chain $\vec S$, $E(N)$ is the extinction time of $N$, and $B(N)$ the total number of births in the chain $N$.
\end{lemma}

To prove the lemma, we construct a Markov chain $(\widehat{\vec S},\widehat{N})$ on the state space $\mathbb{N}^2 \times \mathbb{N}$. Strictly speaking, this will \emph{not} be a coupling of $\vec S$ and $N$, as only the marginal distribution of $\widehat{N}$ will equal the distribution of $N$. However, for each $t \ge 0$, we show how to extract random variables from $\widehat{\vec S}$ whose marginal distribution equals that of $\vec S_t$ almost surely (i.e., with probability 1) for every $t \ge 0$.

\paragraph{The pseudo-coupling.}
We construct the Markov chain $(\widehat{\vec S},\widehat{N})$ as follows.
We set $\widehat{\vec S}_0 = \vec S_0$ and $\widehat{N}_0 = N_0 \ge \min\vec S_0$.
Let $(\xi_t)_{t \ge 0}$ be a sequence of i.i.d.\ random variables distributed
uniformly at random in the unit interval $[0,1)$.
We determine the state of the chain for step $t+1$ inductively:
  \begin{enumerate}
        \item Let $\widehat{N}_t = m$. We set $\widehat{N}_{t+1}$ as follows:

          \begin{enumerate}[label=(\alph*)]
          \item If $\xi_t \in [0, p(m))$, then
              set $\widehat{N}_{t+1} = \widehat{N}_{t} +1 = m+1$.

            \item If $\xi_t \in [1-q(m), 1)$, then set $\widehat{N}_{t+1} = \widehat{N}_t - 1 = m-1$.
              \item Otherwise, set $\widehat{N}_{t+1} = \widehat{N}_t$. (A holding step occurs.)
          \end{enumerate}

        \item Let $\vec S_t = (a,b)$.  If $\min\widehat{\vec S}_t \neq \widehat{N}_t$, we set $\widehat{\vec S}_{t+1} = \widehat{\vec S}_t$. Otherwise, we set $\widehat{\vec S}_{t+1}$ as follows:

	\begin{enumerate}[label=(\alph*)]
	              \item If $\xi_t \in [0, P(a,b))$, then
sample $\widehat{\vec S}_{t+1}$ conditioned on the event that $\widehat{\vec
S}_t = (a,b)$ and the $t$\textsuperscript{th} event is a bad
non-competitive event.

	    \item  If $\xi_t \in [1-Q(a,b), 1)$, then sample $\vec S_{t+1}$
	    conditioned on the event that $\vec S_{t} = (a,b)$ and that the
	    $t$\textsuperscript{th} event is a good competitive
	    interaction.

	      \item Otherwise, if $\xi_t \in [P(a,b), 1-Q(a,b))$, then sample
	      $\vec S_{t+1}$ conditioned on the event that $\vec S_{t} = (a,b)$
	      and that the $t$\textsuperscript{th} event is not  a good
	      competitive
	      interaction or a bad non-competitive event.

  \end{enumerate}

\end{enumerate}

One easily checks that, by construction, the marginal distribution of~$\widehat{N}_t$ is equal to the distribution~$N_t$. Moreover, $\widehat{\vec S}_{t+1} \neq \widehat{\vec S}_{t+1}$ can hold only for steps, where $\widehat{\vec S}_t = \widehat{N}_t$. So $\widehat{\vec S}_t$ does not necessarily have the same marginal distribution as $\vec S_t$. For any $t \ge 0$, define $J_t(\widehat{S})$ to be the number of bad non-competitive events that have occurred in $\widehat{\vec S}$ by time $t$ and $B_t(\widehat{N})$ to be the number of birth events that have occurred in $\widehat{N}$ by time $t$.

We say that an event occurs \emph{almost surely} if it happens with probability 1. For example, we say that $N$ goes extinct almost surely if the extinction time of $N$ is finite with probability 1.

\begin{lemma}\label{lem:coupling:dominates}
If $\min \widehat{\vec S}_0 = \widehat{N}_0$, then
$\min\widehat{\vec S}_t \leq \widehat{N}_t$ and $J_t(\widehat{\vec S}) \le B_t(\widehat{N})$ almost surely for all $t\geq0$.
\end{lemma}
\begin{proof}
  We proceed by induction on $t \ge 0$. The base case $t=0$ is trivial. For the induction step, suppose  $\min \widehat{\vec S}_{t} \leq \widehat{N}_{t}$ and $J_t(\widehat{\vec S}) \le B_t(\widehat{N})$ almost surely.
By the law of total probability, we have
\begin{equation*}
\begin{split}
\Pr[ \min \widehat{\vec S}_{t+1} \leq \widehat{N}_{t+1} ]
 = \\
\Pr[ \min \widehat{\vec S}_{t+1} \leq \widehat{N}_{t+1} \mid \min \widehat{\vec S}_{t} = \widehat{N}_{t} ]
\cdot \Pr[ \min \widehat{\vec S}_{t} = \widehat{N}_{t}]
+
\\ \Pr[ \min \widehat{\vec S}_{t+1} \leq \widehat{N}_{t+1} \mid \min \widehat{\vec S}_{t} < \widehat{N}_{t} ]
\cdot \Pr[ \min \widehat{\vec S}_{t} < \widehat{N}_{t}].
\end{split}
\end{equation*}
The condition $\min\widehat{\vec S}_{t} < \widehat{N}_{t}$ implies $\min\widehat{\vec S}_{t+1} \leq \widehat{N}_{t+1}$ since $\widehat{\vec S}_{t+1} = \widehat{\vec S}_{t}$ and $\widehat{N}_{t+1} \leq \widehat{N}_{t} - 1$ by the transition probabilities of $(\widehat{\vec S}, \widehat{N})$, as given by rules (1) and (2).
For $\min\widehat{\vec S}_{t} = \widehat{N}_{t}$, writing $\widehat{N}_{t}
= m$ and $\widehat{S}_{t} = (a,b)$, we distinguish three cases via the law of
total probability:
\begin{enumerate}[noitemsep]
\item $\xi_t \in [0, P(a,b))$
\item $\xi_t \in [1-q(m), 1)$
\item $\xi_t \in [P(a,b), 1-q(m))$.
\end{enumerate}
In case~(1), by property (D1) of dominating chains, we have $\xi_t \leq P(a,b)\leq p(m)$. Therefore, by rule (1a) we have $\widehat{N}_{t+1} = \widehat{N}_{t} + 1$ and
$\min\widehat{\vec S}_{t+1} \leq \min\widehat{\vec S}_{t} + 1$ by rule (2a).
Combining these inequalities implies $\min \widehat{\vec S}_{t+1} \leq
\widehat{N}_{t+1}$.
In case~(2), since $\xi_t \geq 1- q(m) \geq 1 - Q(a,b)$ by~(D2), we have $\widehat{N}_{t+1} = \widehat{N}_{t} - 1$ by rule (1b) and
$\min \widehat{\vec S}_{t+1} = \min\widehat{\vec S}_{t} - 1$ by rule (2b).
In case~(3), we have $\xi_t \in [P(a,b), 1-q(m)) \subseteq [p(m), 1-Q(a,b))$.
Thus $\widehat{N}_{t+1} \geq \widehat{N}_t$ and $\min\widehat{\vec S}_{t+1} \leq \min\widehat{\vec S}_t$.
Therefore, in all cases, the probability of $\min \widehat{\vec S}_{t+1} \leq \widehat{N}_{t+1}$ is equal to~$1$. This shows the first inequality.

For the second inequality, we proceed in a similar manner and note that a bad non-competitive event can occur in $ \widehat{\vec S}$ only in the case (1). But then by rule (1a) and (2a) we get that $J_{t+1}(\widehat{\vec S}) = J_{t}(\widehat{\vec S}) + 1  \le B_t(\widehat{N}) + 1 = B_{t+1}(\widehat{N})$. This completes the induction.
\end{proof}

For every positive integer~$k$, define the random time~$\tau(k)$ as the $k$\textsuperscript{th} smallest time~$t$ such that $\min\widehat{\vec S}_t = \widehat{N}_t$. This is a stopping time of the Markov chain $(\widehat{\vec S}, \widehat{N})$~\cite[Chapter~2, Definition~7.1]{Bre99}.

\begin{lemma}\label{lemma:delayed-coupling}
Suppose $\widehat{\vec S}_0 = \vec S_0$ and $\min \widehat{\vec S}_0 = \widehat{N}_0$. If $N$ goes extinct almost surely, then
$\tau(k+1) < \infty$ almost surely and the marginal distribution of~$\widehat{\vec S}_{\tau(k+1)}$ equals the distribution of~$\vec S_{k}$ for every $k\geq 0$.
\end{lemma}
\begin{proof}
We proceed by induction on~$k$.
For the base case, by the hypothesis on the initial states, we have $\tau(1) = 0 < \infty$ and
$\widehat{\vec S}_{\tau(1)} = \widehat{\vec S}_0 = \vec S_0$.
This proves the base case.

For the induction step, assume that $k\geq 1$ and that~$\widehat{\vec S}_{\tau(k)}$ and $\vec S_{k-1}$ are equally distributed.
By the strong Markov property~\cite[Chapter~2, Theorem~7.1]{Bre99}, the process $(\widehat{\vec S}_{\tau(k)+t}, \widehat{N}_{\tau(k)+t})_{t\geq 0}$ is a Markov chain with the same transition probabilities as $(\widehat{\vec S}, \widehat{N})$.
By definition of~$\tau(k)$, we have $\min\widehat{\vec S}_{\tau(k)} = \widehat{N}_{\tau(k)}$.
By the definition of the transition probabilities of $(\widehat{\vec S}, \widehat{N})$ for this case, the states~$\widehat{\vec S}_{\tau(k)+1}$ and~$\vec S_{k+1}$ are equally distributed.
We have that
\begin{equation}\label{eq:coupling:between:stopping}
\widehat{\vec S}_{\tau(k)+t} = \widehat{\vec S}_{\tau(k)+1} \textrm{ for all  } 1\leq t < \tau(k+1) -\tau(k) + 1
\end{equation}
almost surely by the definition
of $\tau(k+1)$ and the transition probabilities of $(\widehat{\vec S},
\widehat{N})$.
In particular, if $\tau(k+1) < \infty$, we deduce that~$\widehat{\vec
S}_{\tau(k+1)}$ has the same distribution as~$\widehat{\vec S}_{\tau(k)+1}$, which, in turn, has the same distribution as~$\vec S_k$.
It remains to prove that $\tau(k+1) < \infty$ almost surely.
By \lemmaref{lem:coupling:dominates} and~(\ref{eq:coupling:between:stopping}),
\begin{equation}\label{eq:coupling:onespecies:between:stopping}
\min\widehat{\vec S}_{\tau(k)+1} < \widehat{N}_{\tau(k)+t} \textrm{ for all } 1\leq t < \tau(k+1)-\tau(k) + 1.
\end{equation}
Since $N$ goes extinct almost surely, we have $N_t\to 0$ almost surely. Since $\widehat{N}_t$ and~$N_t$ are equally distributed,  $\lim_{t\to\infty}\widehat{N}_{\tau(k)+t} = 0$ almost surely. Thus, there exists some $t \ge 1$
such that $\min\widehat{\vec S}_{\tau(k)+1} \geq \widehat{N}_{\tau(k)+t}$ almost surely, as $\min\widehat{\vec S}_{\tau(k)+1}\geq 0$.
This contradicts $\Pr[\tau(k+1)=\infty] > 0$, so this probability is hence zero and so $\Pr[\tau(k+1) < \infty] = 1$.
\end{proof}

\paragraph{Proof of the chain domination lemma (\lemmaref{lemma:chain-domination}).}
We are now ready to give the proof of the chain domination lemma.
For the first claim, observe that by \lemmaref{lemma:delayed-coupling}, the marginal distribution of $(\widehat{S}_{\tau(k+1)})_{k \ge 0}$ has equal distribution with the chain $(S_k)_{k \ge 0}$ almost surely.
In particular, since $\tau(k+1) \ge k$ always and the minimum cannot increase after hitting 0, we have that
\begin{align*}
\Pr[\min \vec S_k > 0] &= \Pr[\min \widehat{\vec S}_{\tau(k+1) + 1} > 0] \\
&\le \Pr[\min \widehat{\vec S}_k > 0] \le \Pr[\widehat{N}_k > 0],
\end{align*}
as by \lemmaref{lem:coupling:dominates} the event $\min \widehat{S}_k \le \widehat{N}_k$ holds almost surely. This implies that $T(\widehat{\vec S}) \le E(\widehat{N})$ almost surely.
Therefore, for the consensus time, we get that $T(\vec S)$ is stochastically dominated by $E(N)$, as
\begin{align*}
\Pr[T(\vec S) > k] &\le \Pr[T(\widehat{\vec S}) > k] = \ \Pr[E(\widehat{N}) \ge T(\widehat{\vec S)} > k] \\
&\le \Pr[E(\widehat{N}) > k] = \Pr[E(N) > k].
\end{align*}
This establishes the first claim (a) that $T(\vec S) \preceq E(N)$.

The next step is to show the second claim, i.e., $J(\vec S) \preceq B(N)$, where $J(\vec S)$ is the number of bad non-competitive reactions that occur in $\vec S$ and $B(N)$ is the number of births that occur in chain $N$.
By \lemmaref{lem:coupling:dominates} and \lemmaref{lemma:delayed-coupling}, we have for any $k,x \ge 0$ that
\begin{align*}
\Pr[J_k(\vec S) > x ] &= \Pr[J_{\tau(k+1)}(\widehat{\vec S}) > x] \\
&\le \Pr[B_{\tau(k+1)}(\widehat{N}) \ge J_{\tau(k+1)}(\widehat{\vec S}) > x] \\
&\le \Pr[_{\tau(k+1)}(\widehat{N}) > x] \le \Pr[B(\widehat{N}) > x] \\
&= \Pr[B(N) > x].
\end{align*}
That is, for any $k \ge0$, we have $J_k(\vec S) \preceq B(N)$ and so $J(\vec S) = \sup \{ J_k(\vec S) : k \ge 0 \} \preceq B(N)$.

\subsection{Constructing a dominating chain for competitive Lotka--Volterra systems}

Let be a two-species Lotka--Volterra chain $\vec S$ with $\gamma=0$ and $\alpha_\textrm{min} = \min\{\alpha_0, \alpha_1\} > 0$. We will now define a \emph{nice} birth-death chain $N$ that is a dominating chain for $\vec S$. By the chain domination lemma and results for nice chains given in \sectionref{sec:nice}, this implies that the consensus time of $\vec S$ is $O(n)$ in expectation and with high probability and that the number of bad non-competitive events is $O(\log n)$ in expectation and $O(\log^2n)$ with high probability.
Let $\vartheta = \beta + \delta$, $\alpha = \alpha_0 + \alpha_1$. Define the birth probability function $p$ and death probability function $q$ as
\[
p(m) = \frac{\vartheta }{\alpha m + \vartheta} \quad \textrm{ and  } \quad q(m) = \frac{\alpha_\textrm{min}}{\alpha +  2 \vartheta}
\]
for all $m > 0$ and $p(0)=q(0)=0$. Since $p$ attains its maximum at $p(1)$, we have that $p(m)+q(m) \le p(1) + q(m) \le 1$ for all $m \ge 0$.
Let $(N_t)_{t \ge 0}$ be the birth-death chain defined by $p$ and $q$ and fix $N_0 \ge \min\vec S_0$. Since $p \in O(1/m)$ and $q(m)$ is a positive constant for $m>0$, the chain is nice.

\begin{lemma}\label{lemma:domination}
  The nice birth-death chain $N$ defined by $p$ and $q$ is a dominating chain for $\vec S$.
\end{lemma}

\begin{proof}
  Let $a,b \ge 1$. We need to show that the following conditions hold:
  \begin{enumerate}[noitemsep]
    \item[(D1)] $P(a,b) \le p(\min\{a,b\})$, and
    \item[(D2)] $Q(a,b) \ge q(\min\{a,b\})$.
    \end{enumerate}
  Without loss of generality, assume $a \ge b \ge 1$.
  For the first condition (D1), note that
\[
P(a,b) = \frac{\delta a+ \beta b}{\alpha ab + \vartheta(a+b)} \le \frac{\vartheta a}{\alpha ab + \vartheta a + \vartheta b} \le \frac{\vartheta}{\alpha b + \vartheta + \vartheta b/a} \le \frac{\vartheta}{\alpha b + \vartheta} = p(b).
\]
For the second condition (D2), note that the probability of a good competitive interaction is at least
\begin{align*}
  Q(a,b) \ge
\frac{\alpha_\textrm{min}  \cdot ab} {\alpha ab + \vartheta(a+b)} &= \frac{\alpha_\textrm{min} } {\alpha  + \vartheta(1/b+1/a)} \ge \frac{\alpha_\textrm{min} } {\alpha  + 2\vartheta} = q(b). \qedhere
\end{align*}
\end{proof}

The following is a straightforward consequence of the chain domination lemma (\lemmaref{lemma:chain-domination}), \lemmaref{lemma:domination}, and the results given in \sectionref{sec:nice}.

\begin{theorem}\label{thm:nice-upper-domination}
Let $\vec S$ be a Lotka--Volterra system with $\alpha_\textrm{min}>0$ and $\gamma=0$ and an initial population of size $n>0$. The consensus time $T(\vec S)$ and the number $J(\vec S)$ of bad non-competitive events satisfy
\begin{enumerate}[noitemsep,label=(\alph*)]
\item $\E[T(\vec S)] \in O(n)$ and $T(\vec S) \in O(n)$ with high probability, and
\item $\E[J(\vec S)] \in O(\log n)$ and $J(\vec S) \in O(\log^2n)$ with high probability.
\end{enumerate}
\end{theorem}
\begin{proof}
  By construction, $N$ is a nice birth-death chain and a dominating chain for the two-species chain $\vec S$. Set $N_0 = n$. By \lemmaref{lemma:chain-domination}, we have that $T(\vec S) \preceq E(N)$ and $J(\vec S) \preceq B(N)$. The claim (a) follows now from the fact that $E(N) \in \Theta(n)$ in expectation and with high probability by \lemmaref{lemma:nice-extinction} and \lemmaref{lemma:nice-whp-extinction}. The claim (b) follows from \lemmaref{lemma:nice-expected-births} and \lemmaref{lemma:nice-whp-births}.
\end{proof}

\section{Lotka--Volterra systems with self-destructive competition}\label{sec:self-destructive}

In this section, we show that the threshold for high probability majority consensus lies between $\Omega(\sqrt{\log n})$ and $O(\log^2 n)$ for neutral Lotka--Volterra systems with self-destructive interspecific competition and no intraspecific competition as given by Eq.~(\ref{eq:sd-model}) with $\alpha>0$, $\beta,\delta \ge 0$ and $\gamma=0$.

\paragraph{Demographic noise under self-destructive competition.}
Under self-destructive competition, competitive events cannot change the difference between the two species. Thus, only demographic noise from the individual birth and death reactions play a role.  Let $I(\vec S)$ denote the total number of such events before the chain reaches consensus by time $T(\vec S)$.
Let $F_t = \Delta_{t-1} - \Delta_{t}$. Recall from Eq.~(\ref{eq:noise-sum}) that the probability of reaching majority consensus is given by $\rho(\vec S) = 1 - \Pr[F \le \Delta_0]$, where
\begin{equation}\label{eq:self-destructive-noise}
F = \sum_{t=1}^{T(\vec S)} F_t = \sum_{k=1}^{I(\vec S)} F_{t(k)}
\end{equation}
and $t(k)$ is the time when the $k$th non-competitive event occurs. The equality follows from the fact that  competitive reactions cannot change the gap under self-destructive competition and for any time step $t$ for which a competitive interaction occurs we have $F_t = 0$.

\subsection{Upper bound for self-destructive competition}

\begin{theorem}
  Let $k \ge 0$ be a constant. Suppose $\vec S$ is a Lotka--Volterra system with self-destructive competition and $n = \vec S_0+\vec S_1$. There is a constant $C(k)$ such that if $\Delta_0 > C(k) \log^2 n$, then
  \[
\rho(\vec S) \ge 1-1/n^k.
  \]
\end{theorem}
\begin{proof}
By \theoremref{thm:nice-upper-domination}, the consensus time $T(\vec S)$ is finite with probability one.
From Eq.~(\ref{eq:self-destructive-noise}) we have that the probability of reaching majority consensus is given by $\rho(\vec S) = 1 - \Pr[F \le \Delta_0]$. Now
\[
1 - \rho(\vec S) = \Pr[F \ge \Delta_0] \le \Pr[J(\vec S) \ge \Delta_0] \le 1/n^k,
\]
whenever $\Delta_0 \ge C(k) \log^2 n$, where $C(k)$ is the constant given by \theoremref{thm:nice-upper-domination}.
This is because the $J(\vec S)$ non-competitive reactions that reduce the gap between (current) minority and majority species will not exceed $\Delta_0$ with high probability.
\end{proof}

\subsection{Lower bound for self-destructive competition}

We now show that with self-destructive competition, a Lotka--Volterra system $\vec S$ can fail to reach majority consensus with constant probability if the initial gap is $o(\sqrt{\log n})$.
The idea is to show that when the gap is this small, with at least constant probability, the noise from the individual events will bring the chain into a state $(a,a)$, for some $a>0$, where both species have equal counts. From such a state, a system with identical species (i.e., equal rate parameters), both species have equal probability of winning majority consensus (i.e., going extinct last).
The next lemma expresses this observation more formally.

\begin{lemma}\label{lemma:identical-gap-fail}
For a Lotka--Volterra system $\vec S$ with identical species,
\[
1 - \rho(\vec S) \ge \frac{1}{2} \cdot \Pr[\Delta_t = 0 \textrm{ for some  } t < T(\vec S)].
\]
\end{lemma}

The next lemma applies to Lotka--Volterra systems with self-destructive and non-self-destructive competition and establishes a lower bound on the number~$I(\vec S)$ of individual events.

\begin{lemma}\label{lemma:log-individual-events}
  Let $\vartheta = \beta + \delta$, $\alpha>0$ and $\gamma=0$.
If $\vartheta > 0$, then there exist constants $f,g>0$ such that $I(\vec S) \geq f \log m$ with probability at least $1-1/m^g$, where $m$ is the initial count of the minority species.
\end{lemma}

\begin{proof}
Define $M_t = \min \vec S_t$ for every $t \ge 0$. Note that $M_0 = m$. For every $0 \le k \le m$, define the stopping time $t(k) = \min \{ t : M_t = k\}$ to be the \emph{first} time the chain reaches a state where the (current) minority species has count $k$. Note that $t(k) \le T(\vec S)$ and $T(\vec S)$ is finite with probability 1 by \theoremref{thm:nice-upper-domination}.
Let $X_t$ be the indicator variable denoting if a non-competitive event occurs at time step $t \ge 0$. Since the minimum species count can decrease by at most one in each step,  the sequence $(M_t)_{0\le t < T(\vec S)}$ has to visit every state $m, m-1 \ldots, 1$ at some point before time $T(\vec S)$. Therefore,
\[
I(\vec S) = \sum_{t=1}^{T(\vec S)} X_t \ge \sum_{k=1}^{m} X_{t(k)} = X,
\]
where $X = X_{t(1)} + \cdots + X_{t(m)}$. Thus, to show the claim of the lemma, it suffices to establish that $X \in \Omega(\log m)$ holds with probability $1-1/m$.
Define $p(k) = \vartheta/(\alpha k + 2\vartheta)$ and
observe that when the chain is in state $(a,k)$ or $(k,a)$, where $a \ge k \ge 1$,
the probability a non-competitive event is
\[
\frac{\vartheta(a+k)}{\alpha ak + \vartheta(a+k)} = \frac{\vartheta(1+k/a)}{\alpha k + \vartheta(1+k/a)} \ge \frac{\vartheta}{\alpha k + 2\vartheta} = p(k).
\]
In particular, this implies that for every $0 < k \le m$, we have $\Pr[X_{t(k)} = 1 \mid X_{t(1)}, \ldots, X_{t(k-1)}] \ge p(k)$. Define $Y = Y_1 +  \cdots + Y_m$ as the sum of \emph{independent} Bernoulli random variables with $\Pr[Y_k = 1] = p(k)$. By \lemmaref{lemma:couple-with-independent}, we get that $Y \preceq X$ and so $Y \preceq X \preceq I(\vec S)$.
Since $p(k) \in \Omega(1/k)$, there is some constant $0<c<1$ such that $p(k) \ge c/k$ for all $k$.
Therefore, by linearity of expectation,
\[
\E[Y] = \sum_{k=1}^m \E[Y_k] \ge c \sum_{k=1}^m \frac{1}{k} = c H_m \ge c \ln m.
\]
Since $Y \preceq X \preceq I(\vec S)$, applying the Chernoff bound (\lemmaref{lemma:chernoff}) with $\varepsilon = 1/2$, yields
\begin{align*}
\Pr[I(\vec S) < \frac{c}{2} \ln m] &\le \Pr[I(\vec S) \le (1-\varepsilon) c \ln m] \le \Pr[X \le (1-\varepsilon) c \ln m ]\\
&\le \Pr[Y \le (1-\varepsilon) \cdot \E[Y]] \le \exp\left(-\frac{\varepsilon^2 \cdot \E[Y]}{2}\right) \le 1/m^{c/8}.
\end{align*}
Setting $f = c/2$ and $g = c/8$ now concludes the proof.
\end{proof}

We are now ready to prove that if the initial gap $\Delta_0 \in o(\sqrt{\log n})$, then Lotka--Volterra systems with self-destructive interspecific competition can fail to reach majority consensus with constant positive probability.

\begin{theorem}
  Suppose $\vec S$ is a neutral Lotka--Volterra system with $\beta=\delta$, self-destructive competition $\alpha > 0$, and $\gamma = 0$.
  Let $\varepsilon>0$ be a constant.
  Then there exists a constant $\phi>0$ such that if $\Delta_0 \le \phi \sqrt{\log n}$, then $\rho(\vec S) \le 1/2 + \varepsilon$ for all sufficiently large $n$.
\end{theorem}
\begin{proof}
  Let $\vec S_0 = (m',m)$ be the initial state of the Lotka--Volterra chain, where $m'\ge m>0$ and $m' \leq m + \phi\sqrt{\log n}$ so that the initial gap is $\Delta_0 \leq \phi\sqrt{\log m}$. By \lemmaref{lemma:log-individual-events}, there is a constant $f>0$ such that there are $I(\vec S) \ge K = f\ln m$ individual events before the chain reaches consensus with probability at least $1-1/m^g$.
  Let $t(k)$ be the time exactly $k>0$ individual events have happened. As before, the probability of reaching majority consensus is given by $\rho(\vec S) = 1 - \Pr[F \ge \Delta_0]$, where
\[
F = \sum_{t=1}^{T(\vec S)} F_t = \sum_{k=1}^{I(\vec S)} F_{t(k)},
\]
as competitive reactions cannot change the gap under self-destructive competition. For $1 \le k \le K$, let $X_k = F_{t(k)}$ conditioned on the event that $I(\vec S) \ge K$. By assumption $\beta = \delta$, so we have
\begin{align*}
\Pr[X_k = 1 \mid \vec S_{t(k)} = (a,b)] &= \frac{\delta a +  \beta b }{(\beta+\delta)(a+b)} = \frac{1}{2} \\
&= \frac{\beta a +  \delta b }{(\beta+\delta)(a+b)} \\
&= \Pr[X_k = -1 \mid  \vec S_{t(k)} = (a,b)].
\end{align*}
Thus, conditioned on $I(\vec S) \ge K$, the random variables $X_1, \ldots, X_K$ are (conditionally) independent. Moreover, they have mean $\E[X_k] = 0$ and variance $\Var[X_k] = 1$.
Let $X = X_1 + \cdots + X_K$. By \lemmaref{lemma:clt}, for any constant $\varepsilon > 0$, we can choose a constant $\theta>0$ such that
\[
\Pr[X > \theta \sqrt{K} \mid I(\vec S) \ge K] \ge 1 - \varepsilon/2
\]
for all sufficiently large $K$.
Set $\phi = 3\theta \sqrt{f}$.
For large enough~$n$, we have $n \le 3m$ and hence $\phi\sqrt{\log n} \le \theta\sqrt{K}$.
Conditioned on the event $I(\vec S) \ge K$, we have that with probability $1-\varepsilon/2$ the two-species chain reaches some state $(a,a)$, where $a>0$. From this point onward, by \lemmaref{lemma:identical-gap-fail}, the probability that the majority species not winning is at least $1/2$, as the system is neutral.
Therefore,
when~$n$, and thus~$m$, is sufficiently large, the probability to reach majority consensus is at most
\begin{align*}
\rho(\vec S) &\le \Pr[I(\vec S) < K] + \Pr[X \le \theta \sqrt{K} \mid I(\vec S) \ge K] + 1/2
\\ & \leq 1/m^g + \varepsilon/2 + 1/2 \leq 1/2 + \varepsilon.\qedhere
  \end{align*}
\end{proof}

\section{Lotka--Volterra systems with non-self-destructive competition}\label{sec:non-self-destructive}

In this section, we show that the threshold for high probability majority consensus lies between $\Omega(\sqrt{n})$ and $O(\sqrt{n\log n})$ for Lotka--Volterra systems with non-self-destructive interspecific competition and no intraspecific competition.
That is, we consider the model given by the reactions
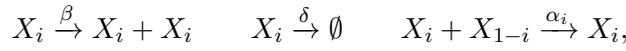
\[
X_i \xrightarrow{\beta} X_i + X_i \qquad X_i \xrightarrow{\delta} \emptyset \qquad X_{i} + X_{1-i} \xrightarrow{\alpha_i} X_i,
\]
where $\alpha_i,>0$, $\beta,\delta \ge 0$ and $i \in \{0,1\}$. For the upper bound, we allow for non-symmetric competition $\alpha_0 \neq \alpha_1$. That is, the minority species can be a better competitor than the majority species.

Unlike in the previous section with self-destructive competition, when we have non-self-destructive competition, also the competition events give arise to noise which influences the  gap between the two species. In particular, there will be $\Theta(m)$ competition events in a system with initial minority of size $m$, so the noise term $F$ will essentially be $\Omega(\sqrt{m})$ and at most $O(\sqrt{m \log m})$ as we will see.

\subsection{Upper bound}

\begin{theorem}
    Suppose $\vec S$ is a Lotka--Volterra system with \emph{non-self-destructive} interspecific competition and $\vec S_0 + \vec S_1 = n$. For any constant $k \ge 0$, there is a constant $\theta(k)$ such that
  if $\Delta_0 > \theta(k) \cdot \sqrt{n \log n}$, then $\rho(\vec S) \ge 1-1/n^k$.
\end{theorem}
\begin{proof}
 Let $I(\vec S)$ be the number of individual events before the chain reaches consensus and $K(\vec S)$ be the number of competitive events before the chain reaches consensus.
 Clearly, $T(\vec S) = I(\vec S) + K(\vec S)$.
 Now
\[
F = \sum_{t=1}^{T(\vec S)} F_t = X + Y, \textrm{ where  } X = \sum_{i=1}^{I(\vec S)} X_i \textrm{ and } Y=\sum_{i=1}^{K(\vec S)} Y_i,
\]
where $F$ and $F_t$ are defined as before in Eq.~(\ref{eq:noise-sum}), and $X_i$ is the outcome of the $i$th non-competitive event and $Y_i$ the outcome of the $i$th competitive event. That is, $X_i = 1$ if the gap between the initial minority and majority species decreases during the $i$th non-competitive event before the chain reaches consensus, and $X_i=-1$ otherwise. Similarly, $Y_i = 1$ if the gap between the initial minority and majority species decreases during the $i$th competitive event before the chain reaches consensus, and $Y_i=-1$ otherwise.

Recall that, as before, the probability of majority consensus is $\rho(\vec S) = \Pr[F < \Delta_0]$.
By \theoremref{thm:nice-upper-domination}(b), we have that $\Pr[J(\vec S) > \theta \sqrt{n \log n} ] \le \Pr[J(\vec S) > \theta \log^2 n ]\le 1/n^{k+1}$ for some constant $\theta>0$.
That is, there are with high probability $O(\log^2  n)$ events that decrease the gap between the (current) majority and minority species before the chain reaches consensus.

We now show that $Y$ is $O(\sqrt{n \log n})$ with high probability.
By \theoremref{thm:nice-upper-domination}(a), we have that $\Pr[K(\vec S) \ge cn] \le \Pr[T(\vec S) \ge cn] \le 1/n^{k+1}$ for some constant $c>0$ and all sufficiently large $n$.
Conditioned on the event $K(\vec S) = \ell$, the sum $Y$ is a sum of $\ell$ (conditionally) independent random variables taking values between $[-1, 1]$.
Thus, by the law of total probability and applying Hoeffding's inequality with $t = \sqrt{(k+1)cn \ln n}$, we get
\begin{align*}
\Pr[Y \ge t] &= \sum_{\ell=0}^\infty \Pr[K(\vec S) = \ell] \cdot \Pr[ Y \ge t  \mid K(\vec S) = \ell] \\
&\le \sum_{\ell=t}^{cn} \Pr[K(\vec S) = \ell] \cdot \Pr[ Y \ge t  \mid K(\vec S) = \ell] + \sum_{\ell > cn} \Pr[K(\vec S) = \ell] \\
&\le 2 \cdot \sum_{\ell=t}^{cn} \Pr[K(\vec S) = \ell] \cdot  \exp \left(-\frac{2t^2 }{\ell}\right) + \Pr[K(\vec S) > cn] \\
&\le 2cn \cdot \exp \left(-\frac{2t^2 }{cn}\right) + 1/n^{k+1}
\le \frac{2c}{n^{2k+1}} + 1/n^{k+1} \le 2/n^{k+1}
\end{align*}
for all sufficiently large $n$. Note that $t +  \theta \sqrt{n \log n} \le \theta(k) \cdot \sqrt{n \log n}$ holds for some sufficiently large constant $\theta(k) \ge \theta'$ depending only on $k$.
Therefore, if $\Delta_0 \ge \theta(k) \sqrt{n\log n}$, then
\begin{align*}
1 - \rho(\vec S) &= \Pr[F \ge \Delta_0] \le \Pr[J(\vec S) + Y \ge \Delta_0] \\
&\le \Pr[J(\vec S) \ge \theta \cdot \sqrt{n \log n}] + \Pr[Y \ge t] \le 3/n^{k+1} \le 1/n^k
\end{align*}
for all sufficiently large $n$, proving the claim.
\end{proof}

\subsection{Lower bound}

\begin{theorem}
Let $\varepsilon > 0$ be a constant. Suppose $\vec S$ is a neutral Lotka--Volterra system with $\beta=\delta$, non-self-destructive competition and $\gamma = 0$. Then there exists a constant $\phi>0$ such that if $\Delta_0 \le \phi \sqrt{n}$, then then $\rho(\vec S) \le 1/2 + \varepsilon$ for all sufficiently large $n>0$.
\end{theorem}
\begin{proof}
Suppose $\vec S_0 = (m + \Delta_0, m)$ is the initial state of the chain $\vec S$ with $m > 0$ and $\Delta_0 \le  \phi \sqrt{m}$, where $\phi>0$ is a constant we fix later. Now the total initial population size is $n = 2m + \Delta_0 \in \Theta(m)$ and the initial gap is $\Delta_0 \in O(\sqrt{n})$.

Let $X_1, \ldots, X_m$ denote the outcomes of the first $m$ events, where $X_i = -1$ if $\Delta_i$ decreases in the event and $X_i = 1$ otherwise. By definition, the random variables are independent and identically distributed: $\Pr[X_i = 1] = 1/2 = \Pr[X_i = -1]$ since the competition rates satisfy $\alpha_0 = \alpha_1 > 0$, as the system is neutral (i.e., the species have identical rate parameters) and $\beta = \delta$.
Thus,  for each $1 \le i \le m$, we have $\E[X_i] = 0$ and $\Var[X_i] = 1$.
  Letting $X = X_1 + \cdots + X_m$ and applying \lemmaref{lemma:clt}, we get
\[
\Pr[X \ge \Delta_0] = \Pr[X \ge \theta \sqrt{m}] \ge 1 - \varepsilon
\]
Thus, before reaching consensus, the chain reaches some state $(a,a)$, where $a>0$, with probability at least $1-\varepsilon$.
From this point onward, by \lemmaref{lemma:identical-gap-fail}, the probability that the majority species losing is at least $1/2$, as the species have identical birth, death and competition rates. Therefore,
when $m$ is sufficiently large, the probability to reach majority consensus is at most
\begin{equation*}
\rho(\vec S) \le \Pr[X < \theta \sqrt{m}] + 1/2
\le 1/2 + \varepsilon.
\qedhere
\end{equation*}
\end{proof}

\section{Lower bounds for systems with intraspecific competition}\label{sec:intra-bounds}

In contrast to our earlier results for systems without intraspecific competition, we now show that Lotka--Volterra systems with intraspecific competition (i.e., $\gamma>0$) can be \emph{much} worse signal amplifiers than systems with only interspecific competition. Namely, the probability of reaching majority consensus can be very small.

We show that when the system has both intraspecific and interspecific competition, that is, $\alpha, \gamma>0$ holds, then the threshold for majority consensus can be $\Omega(n)$. In systems with only intraspecific competition, that is, $\gamma>\alpha = 0$, we show that no such threshold exists.

\paragraph{The models with intra- and interspecific competition.}
In this section, we consider the neutral model with self-destructive interspecific and intraspecific competition given by the reactions
\[
X_i \xrightarrow{\beta} 2X_i \qquad X_i \xrightarrow{\delta} \emptyset \qquad X_{i} + X_{1-i} \xrightarrow{\alpha} \emptyset \qquad X_i + X_{i} \xrightarrow{\gamma} \emptyset,
\]
where $i \in \{0,1\}$ and $\alpha,\gamma,\beta,\delta\ge0$. 

\subsection{Systems with both intra- and interspecific competition}\label{sec:both-comp}

\paragraph{Self-destructive competition.}
We consider first the case of self-destructive competition.
We now establish the following result.

\begin{theorem}\label{thm:sd-intra}
  Let $\vec S$ be a Lotka--Volterra chain with self-destructive competition and initial configuration $\vec S_0 = (a,b)$, where $a \ge b>0$, $\alpha=\gamma \ge 0$ and $\beta, \delta \ge 0$. If the chain reaches consensus with probability 1, then
  \[
\rho(\vec S) = \frac{a}{a+b}.
\]
\end{theorem}

The above immediately implies that the $\Delta_0 = n-1$ is the only positive gap that can guarantee majority consensus under self-destructive competition with probability at least $1-1/n$ if $\gamma = \alpha$ holds.
The above result is a generalization of a similar result by Andaur et al.~\cite{andaur2021reaching}, who established the same bound for the case $\alpha=\gamma=0$ two \emph{independent}, non-competing species, to systems with two competing \emph{non-independent} species.
Let us write $\rho(a,b) = \rho(\vec S)$ and $T(a,b) = T(\vec S)$ for a chain $\vec S$ with initial state $(a,b)$. Let $P$ be the transition probability matrix of the chain $\vec S$. By the law of total probability and the fact that the chain is Markovian, we have that $\rho$ satisfies the recurrence
\begin{equation}\label{eq:rho-recurrence}
\rho(a,b) = \sum_{x,y \ge 0} P((a,b), (x, y)) \cdot \rho(x,y).
\end{equation}
with boundary conditions $\rho(a,0) = 1$ for all $a>0$ and $\rho(0,b) = 0$ for all $b \ge 0$. A straightforward but tedious calculation shows if $\alpha=\gamma$, then $\rho(a,b) = a/(a+b)$ is a solution for the recurrence. Furthermore, assuming that the chain reaches consensus with probability 1, we can show that the solution is unique.

\begin{lemma}\label{lemma:rho-sol}
  The function $a/(a+b)$ is a solution for the recurrence given in Eq~(\ref{eq:rho-recurrence}).
\end{lemma}
\begin{proof}
  We now give the proof in the case of self-destructive competitive interactions, i.e., a competitive interaction removes both participating individuals from the system.  There are in total seven types of events: for both species, there are individual birth and death events, which change the population count by one, and \emph{intraspecific} competitive events, where two individuals of the same species interact and are removed. Finally, there are the  \emph{interspecific} competition events that remove one individual of both species. With this in mind, we can rewrite the recurrence as
  \[
\rho(a,b) = \sum_{(j,k) \in \mathcal{T}_1 \cup \mathcal{T}_2} P((a,b), (a+j,b+k)) \cdot \rho(a+j,b+k),
  \]
where
\[
\mathcal{T}_1 = \{ (1,0), (0,1), (-1,0), (0,-1) \} \quad \textrm{ and } \quad \mathcal{T}_2 = \{ (-2,0), (0,-2), (-1,-1) \}.
\]
Let $r_i(a,b)$ be the conditional probability that the first species wins when starting from state $(a,b)$ conditioned on the event that the next step involves $i$ individuals. Let $\mathcal{F}$ denote the event that the next step involves exactly one individual and observe that
\[
\rho(a,b) = \Pr[\mathcal{F}] \cdot r_1(a,b) + (1-\Pr[\mathcal{F}]) \cdot r_2(a,b).
\]
We now show that $r_i(a,b) = a/(a+b)$ for both $i \in \{1,2\}$, which then implies that $a/(a+b)$ is a solution for the recurrence. First, a simple calculation shows that
\[
r_1(a,b) = \frac{\beta a \cdot \rho(a+1, b) + \delta a \cdot \rho(a-1, b) + \beta b \cdot \rho(a, b+1) + \delta b \cdot \rho(a, b-1)}{(\beta+\delta)(a+b)}
\]
Plugging in $\rho(a,b) = a/(a+b)$ yields
\begin{align*}
r_1(a,b) &= \frac{1}{(\beta+\delta)(a+b)} \left( \beta a \cdot \frac{a+1}{a+b+1} + \delta a \cdot \frac{a-1}{a+b-1}  + \beta b \cdot \frac{a}{a+b+1}
 + \delta b \cdot \frac{a}{a+b-1}\right) \\
 &= \frac{a}{a+b}.
\end{align*}
On the other hand, we have
\begin{align*}
  r_2(a,b) &= \frac{\alpha ab \cdot \rho(a-1,b-1) + \gamma a(a-1)/2 \cdot \rho(a-2,b) + \gamma b(b-1)/2 \cdot \rho(a, b-2)}{\alpha ab + \gamma a(a-1)/2 + \gamma b(b-1)/2} \\
  &= \frac{2ab \cdot \rho(a-1,b-1) + a(a-1) \cdot \rho(a-2,b) + b(b-1) \cdot \rho(a, b-2)}{2ab + a(a-1) + b(b-1)}
\end{align*}
using the assumption $\alpha=\gamma$. Using the identity $(a+b)(a+b-1) = 2ab + a(a-1) + b(b-1)$ twice, and plugging in $\rho(a,b)=a/(a+b)$ we get that
\begin{align*}
  r_2(a,b) &= \frac{a}{(a+b)(a+b-1)} \cdot \frac{2(a-1)b + (a-1)(a-2)+b(b-1)}{(a+b-2)} = \frac{a}{a+b}. \qedhere
\end{align*}
\end{proof}

\begin{lemma}\label{lemma:rho-unique}
  If the chain reaches consensus with probability one, then
the recurrence given by Eq.~(\ref{eq:rho-recurrence}) has a unique solution.
\end{lemma}
\begin{proof}
   Define $d(a,b) = f(a,b) - g(a,b)$, where $f$ and $g$ are two solutions for the recurrence. We will prove the claim by showing that $|d(a,b)| < \varepsilon$ for any $\varepsilon > 0$. Under the assumption that the chain reaches consensus with probability 1, we have that the consensus time $T(a,b)$ satisfies
\[
\lim_{k \to \infty} \Pr[T(a,b) > k] = 0.
\]
Therefore, for any $\varepsilon > 0$, there exists $k$ such that $\Pr[T(\vec S) > k] < \varepsilon$. Now for any $k>0$, $d(a,b)$ satisfies the recurrence
\begin{align*}
d(a,b) &= \sum_{x,y \ge 0} P^k((a,b), (x,y)) \cdot d(x,y) = \sum_{x,y>0} P^k((a,b), (x,y)) \cdot d(x,y),
\end{align*}
where the last equality follows fact that $d(0,x) = d(x,0)=0$ for any $x \ge 0$ due to the boundary conditions of the recurrence for $\rho$.
Since $|d(x,y)| \le 1$, we have that
\[
|d(a,b)| \le \sum_{x,y>0} P^k((a,b), (x,y)) \cdot | d(x,y) | \le \Pr[T(a,b) > k] < \varepsilon. \qedhere
\]
\end{proof}

Now the two lemmas imply \theoremref{thm:sd-intra}.

\paragraph{Non-self-destructive competition.}
We next consider the case of non-self-destructive competition.
The model with non-self-destructive competition is given by
\[
X_i \xrightarrow{\beta} 2X_i \qquad X_i \xrightarrow{\delta} \emptyset \qquad X_{i} + X_{1-i} \xrightarrow{\alpha_i} X_i \qquad X_i + X_{i} \xrightarrow{\gamma_i} X_i,
\]
where we assume that $\alpha= \alpha_0+\alpha_1 = 2\alpha_0$ and $\gamma = \gamma_0 + \gamma_1 = 2\gamma_0$, i.e., that both species have equal rate parameters.

\begin{theorem}
Let $\vec S$ be a neutral Lotka--Volterra system with non-self-destructive competition and $\vec S_0 = (a,b)$, where $a \ge b > 0$, and $\gamma = 2 \alpha$. If the chain reaches consensus with probability 1, then
\[
\rho(\vec S) = \frac{a}{a+b}.
\]
\end{theorem}
\begin{proof}
With non-self-destructive competition, the recurrence for $\rho$ simplifies to
  \[
\rho(a,b) = \sum_{(j,k) \in \mathcal{T}} P((a,b), (a+j,b+k)) \cdot \rho(a+j,b+k),
  \]
where $\mathcal{T} = \{ (1,0), (0,1), (-1,0), (0,-1) \}$, as the counts of species can change by at most one under non-self-destructive competition. We now show that $a/(a+b)$ is a solution for the recurrence $\rho(a,b)$ also under non-self-destructive competition.
  We proceed as in the proof of \lemmaref{lemma:rho-sol} and compute the conditional probabilities $r_i$ for $i\in\{0,1\}$.
The analysis for the function $r_1(a,b) = a/(a+b)$ is identical as in \lemmaref{lemma:rho-sol}. The case for the function $r_2(a,b)$ is different.
  If the chain is in state $(a,b)$, conditioned on the event that the next step has a competitive interaction, the probability that the first species with count $a$ decreases is
  \[
\frac{\alpha_0 ab + \gamma_0 a(a-1)/2}{\alpha ab + \gamma_0 a(a-1)/2 + \gamma_1 b(b-1)/2}  = \frac{ab + a(a-1)}{2ab + a(a-1) + b(b-1)} = \frac{a (a+b-1)}{(a+b)(a+b-1)} = \frac{a}{a+b}.
\]
Thus, we get that
\begin{align*}
r_2(a,b) &= \frac{a}{a+b} \cdot \rho(a-1, b) + \left(1-\frac{a}{a+b}\right) \cdot \rho(a,b-1) \\
&= \frac{a}{a+b} \cdot \frac{a-1}{a+b-1} + \frac{b}{a+b} \cdot \frac{a}{a+b-1} = \frac{a}{a+b} \cdot \frac{a+b-1}{a+b-1} = \frac{a}{a+b}.
\end{align*}
This implies that $\rho(a,b) = a/(a+b)$ is a solution for the recurrence $\rho$. By \lemmaref{lemma:rho-unique}, the solution is unique under the assumption that the chain reaches consensus with probability 1.
\end{proof}

\subsection{Systems with intraspecific competition only}\label{sec:intra-only}

Finally, to complete the picture, we consider the limiting case with intraspecific competition only ($\gamma>0$ and $\alpha=0$). As expected, this case is even worse: for \emph{any} initial configuration, where both species are present, the system fails to reach majority consensus with constant probability.

In the case $\alpha=0$, it is convenient to consider the continuous-time version of the Lotka--Volterra chain, as in this case the two chains for the individual species are independent, continuous-time single-species processes.
We first introduce the following lemma, which readily follows from standard absorption time results on continuous-time birth-death chains~\cite{karlin1975first}. For example, the next result can be proven using a straightforward adaption of the proof of given by Andaur et al.~\cite[Lemma 6]{cho2021distributed}.

\begin{lemma}\label{lemma:continuous-extinction}
Let $M$ be a \emph{continuous-time} birth-death chain with 0 as the unique absorbing state. If the birth rate $\kappa(n)$ and death rate $\mu(n)$ satisfy $\kappa(n) \in \Theta(n)$ and $\mu(n) \in \Theta(n^2)$, then the mean absorption time from any state $m$ is $O(1)$.
\end{lemma}

\begin{theorem}
  Let $\vec S$ be a Lotka--Volterra chain with intraspecific competition and without interspecific competition, i.e.,  $\alpha=0$ and $\gamma>0$. Then $\vec S$ fails to reach majority consensus with at least constant probability from any starting state $\vec S_0 = (a,b)$, with $a \ge b>0$.
\end{theorem}
\begin{proof}
For the proof of this result, it is convenient to consider the continuous-time version $\vec X = (X_0, X_1)$ of the Lotka--Volterra chain instead of the discrete-time jump chain $\vec S$. Since $\alpha=0$, the two continuous-time chains $X_0$ and $X_1$ are independent.
Let $T_i$ denote the time until the species $i \in \{0,1\}$ goes extinct, i.e., the absorption time of the chain $X_i$.
First, note that $\E[T_i] \in O(1)$. For non-self-destructive competition, this follows from \lemmaref{lemma:continuous-extinction}, as the birth and death rates for the chains $X_i$ satisfy $\kappa(n) = \beta n$ and $\mu(n) = \delta n + \gamma_i n(n-1)/2$. For self-destructive competition, the extinction time is also $O(1)$, as death events under self-destructive competition also always decrease the population count by at least one.

Let $Y(\lambda)$ denote the exponentially distributed random variable with mean $1/\lambda$. Observe that for the species $i$ to go extinct, its chain transitions to state 0 via either

\begin{enumerate}
\item state 1, i.e., the last event before extinction is an individual death event, which happens at rate $\delta$ so the time between the second last and last transition is $Y(\delta)$, and

\item state 2, i.e., the last event before extinction is a (self-destructive) competitive interaction, which happens at rate $\gamma$ so the time between the second last and last transition is $Y(\gamma)$.
\end{enumerate}
Note that if the system does not have self-destructive competition, then the second case cannot occur. Let $X$ be an indicator random variable for whether the first case occurs and set $\lambda = \max \{ \alpha, \delta \}$. Now
\begin{align*}
  \Pr[T_i \ge x] &= \Pr[X=1] \cdot \Pr[T_i \ge x \mid X=1] + \Pr[X=0] \cdot \Pr[T_i \ge x \mid X=0]  \\
  &\ge \Pr[X=1] \cdot \Pr[Y(\delta) \ge x] + \Pr[X=0] \cdot \Pr[Y(\gamma) \ge x] \\
  &\ge \Pr[Y(\lambda) \ge x] = \exp(-\lambda x),
\end{align*}
where the last equality is given by the CDF of the exponential distribution. By Markov's inequality,
\[
\Pr[T_0 < 2 \cdot \E[T_0]] \ge 1 - \Pr[T_0 \ge 2 \cdot \E[T_0] \ge 1/2.
  \]
  Since $T_0$ and $T_1$ are independent random variables, we have
  \begin{align*}
    \Pr[T_0 < T_1] &\ge \Pr[T_0 < 2 \cdot \E[T_0]] \cdot \Pr[T_1 \ge 2 \cdot \E[T_0]] \ge \frac{1}{2} \cdot \exp(-2\lambda \cdot \E[T_0]),
  \end{align*}
  where the last term is lower bounded by some positive constant $\theta > 0$ since $\lambda, \E[T_0] \in O(1)$. That is, the minority species survives longer with probability at least~$\theta$.
\end{proof}

\section*{Acknowledgements}
This research was supported by the ANR projects DREAMY (ANR-21-CE48-0003) and COSTXPRESS (ANR-23-CE45-0013).

\bibliographystyle{ACM-Reference-Format}
\bibliography{references}

\section{Omitted proofs}\label{apx:omitted}

\subsection{Proof of \lemmaref{lemma:clt}}\label{apx:clt}

\begin{proof}
By the Central Limit Theorem, we have that the random variable $X/\sqrt{n}$ converges in distribution to the standard normal distribution, i.e., for any $x \in \mathbb{R}$, we have
\[
\lim_{n \to \infty} \Pr\left[\frac{X}{\sqrt{n}} \le x \right] = \Phi(x),
\]
where $\Phi$ is the CDF of the standard normal distribution.
In particular, for any $\kappa > 0$ and $\theta>0$, there exists some $n_0>0$ such that
\[
\Pr\left[X > \theta \sqrt{n} \right] = 1 - \Phi(\theta) - \kappa \ge 1 - \frac{e^{-\theta^2/2}}{\sqrt{2\pi} (\theta+1)} - \kappa \ge \varepsilon
\]
for all $n \ge n_0$, where in the second inequality we used the bound $\Phi(\theta) \le \frac{e^{-\theta^2/2}}{\sqrt{2\pi}(1+\theta)}$ for $\theta > 0$. The claim now follows by choosing $\theta$ and $\kappa$ to be sufficiently small constants and $n$ sufficiently large.
\end{proof}

\subsection{Proof of \lemmaref{lemma:couple-with-independent}}\label{apx:couple-with-independent}

\begin{proof}
For the case (a), we construct a coupling $(\widehat{X}, \widehat{Y})$ of $X$ and $Y$ such that $\widehat{X} \le \widehat{Y}$. Let $\xi_1, \ldots, \xi_n$ be i.i.d.\ random variables sampled uniformly from the unit range $[0,1)$. We construct the coupling inductively. For each step $1 \le i \le n$, the coupling uses the following rules:
    \begin{enumerate}
      \item   Let $p_i = \Pr[Y_i]$. If $\xi_i \in [0, p_i)$, then set $\widehat{Y}_i = 1$. Otherwise, set $\widehat{Y}_i = 0$.

      \item Let $P_i = \Pr[X_i \mid X_1 = \widehat{X}_1, \ldots, X_{i-1} = \widehat{X}_{i-1}]$. If $\xi_i \in [0, P_i)$, then set $\widehat{X}_i = 1$. Otherwise, set $\widehat{X}_i = 0$.
  \end{enumerate}
By construction, the distribution of $\widehat{X} = \widehat{X}_1 + \cdots + \widehat{X}_n$ is the same as the distribution of $X$  and  $\widehat{Y} = \widehat{Y}_1 + \cdots + \widehat{Y}_n$ has the same distribution as $Y$. Moreover, $\widehat{X}_i \le \widehat{Y}_i$, as by assumption we have $[0, P_i) \subseteq [0, p_i)$ for all $1 \le i \le n$. In particular, $\Pr[\widehat{X}_i \ge x] \le \Pr[\widehat{Y}_i \ge x]$, so we get that
\[
\Pr[Y \ge x] = \Pr[\widehat{Y} \ge x] \ge \Pr[\widehat{X} \ge x] =
\Pr[X \ge x],
\]
so $X \preceq Y$. The proof for the case (b) is symmetric and follows analogously by observing that
$[0,p_i) \subseteq [0,P_i)$ implies that $\widehat{Y}_i \le \widehat{X}_i$ for each $1 \le i \le n$.
\end{proof}

\end{document}